\documentclass[a4paper,UKenglish,cleveref, autoref, thm-restate,notab,cleveref]{lipics-v2021}

\usepackage{float}
\newfloat{model}{h}{lom}   
\floatname{model}{Model}     

\title{Strong ILP Formulations for the $p$-Regions Problem} 

\author{Daniel Faber}{University of Bonn, Germany }{dfaber@uni-bonn.de}{https://orcid.org/0009-0002-1896-9402}{}

\author{Jan-Henrik Haunert}{University of Bonn, Germany}{haunert@igg.uni-bonn.de}{https://orcid.org/0000-0001-8005-943X}{}

\author{Petra Mutzel}{University of Bonn, Germany \and Lamarr Institute, Bonn, Germany}{pmutzel@uni-bonn.de}{https://orcid.org/0000-0001-7621-971X}{}

\authorrunning{D. Faber, J.-H. Haunert, P. Mutzel} 

\Copyright{Daniel Faber, Jan-Henrik Haunert, and Petra Mutzel} 

\ccsdesc[500]{Theory of computation~Discrete optimization} 

\keywords{p-regions problem, connected graph partitioning, area aggregation, integer linear programming, branch-and-cut} 

\category{} 

\relatedversion{Full version on arXiv: \url{https://arxiv.org/abs/2607.04886}} 

\supplement{Code available under \url{https://zenodo.org/records/21275790}}

\funding{This research was partially funded by the Deutsche Forschungsgemeinschaft (DFG, German Research Foundation) under grant FOR-5361 -- 459420781.}

\acknowledgements{We want to thank Michael Kaibel for his helpful remarks.}

\nolinenumbers 

\EventEditors{John Q. Open and Joan R. Access}
\EventNoEds{2}
\EventLongTitle{42nd Conference on Very Important Topics (CVIT 2016)}
\EventShortTitle{CVIT 2016}
\EventAcronym{CVIT}
\EventYear{2016}
\EventDate{December 24--27, 2016}
\EventLocation{Little Whinging, United Kingdom}
\EventLogo{}
\SeriesVolume{42}
\ArticleNo{23}

\usepackage{todonotes}

\usepackage[]{algorithm}    
\usepackage{algpseudocode}     
\usepackage{placeins}
\usepackage{amsmath}
\usepackage{amssymb}
\usepackage{subcaption}
\usepackage{hyperref}
\usepackage{tikz}
\usetikzlibrary{arrows.meta, positioning}
\usepackage{xspace}
\usepackage{booktabs}

\newcommand{\diedge}[1]{A({#1})}
\newcommand{\uvsep}{\text{$u,v$-separator}\xspace}
\newcommand{\tree}{\texttt{Tree}\xspace}
\newcommand{\treeplus}{\texttt{Tree+}\xspace}
\newcommand{\flow}{\texttt{Flow}\xspace}
\newcommand{\ER}{\texttt{ER}\xspace}
\newcommand{\ERS}{\texttt{ER-S}\xspace}
\newcommand{\ERSBase}{\texttt{ER-S-Base}\xspace}
\newcommand{\ERSTree}{\texttt{ER-S-Tree}\xspace}
\newcommand{\Size}{\emph{Size}\xspace}
\usepackage{comment}

\begin{document}

\maketitle
\begin{abstract}
Regionalization is a fundamental task in spatial analysis that seeks to partition a larger area - such as a country - into smaller regions that are homogeneous with respect to a given attribute. A popular model for regionalization is the $p$-regions problem, in which regions are formed by grouping the areas of an input planar subdivision. Given the subdivision's adjacency graph $G$ and pairwise dissimilarities between vertices, the goal is to partition $G$ into a fixed number $p$ of connected subgraphs, such as to minimize the sum of dissimilarities over all vertex pairs in the same subgraph. The problem is NP-hard and even small instances are difficult to solve to provable optimality. 

In this paper, we present the new ILP model \ERS for the $p$-regions problem, exploiting a connection between the $p$-regions objective and the $k$-partitioning problem. Furthermore, we strengthen the known ILP model \tree with a new type of subtour elimination inequality specific to the $p$-regions problem. Combining \ERS and the strengthened version of \tree yields the model \ERSTree, which dominates the state-of-the-art models in polyhedral strength. This theoretical advantage is reflected in its superior performance in our experimental evaluation.  In particular, the new models \ERS and \ERSTree enable the solution of problem instances for major European countries that were previously intractable.
\end{abstract}

\section{Introduction}
\begin{figure}
\begin{minipage}{0.44\linewidth}
\centering
\includegraphics[width=\linewidth]{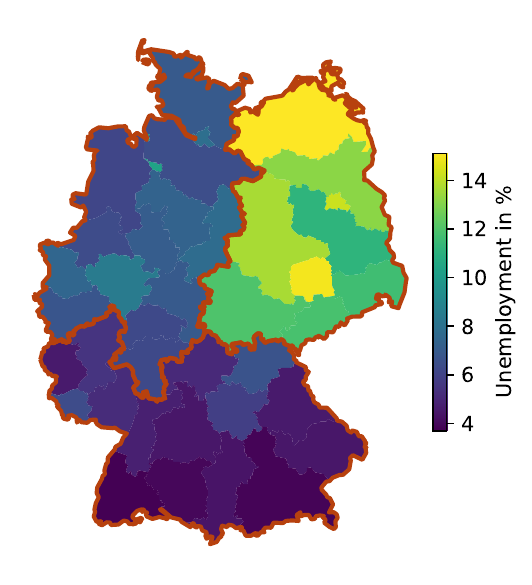}
\\
\subcaption{$k = 3$} 
\end{minipage}
\hfill
\begin{minipage}{0.44\linewidth}
\centering
\includegraphics[width=\linewidth]{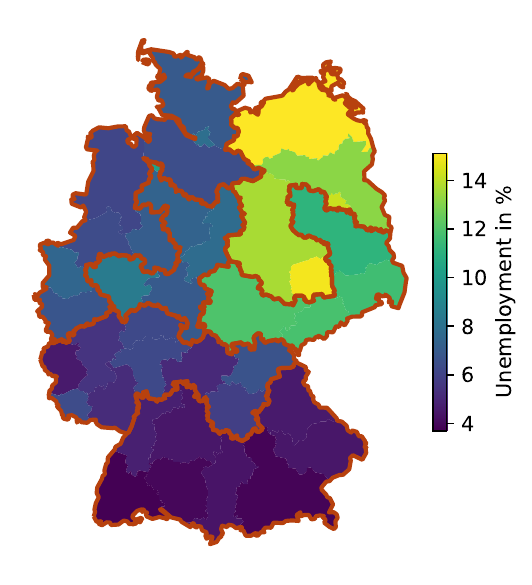}
\\
\subcaption{$k = 6$}
\end{minipage}
\caption{Regionalization solutions for the NUTS-2 subdivison of Germany for $k \in \{3,6\}$ }
\label{fig:germany-part}
\end{figure}
The aggregation of spatial data is a core topic in geoinformation science, and has given rise to many computationally challenging problems. A common restriction in these aggregation problems is regional connectivity: Given a set of geographic areas, the goal is to group similar (with respect to a given attribute) areas into a smaller number of regions, where each region needs to be contiguous. There are numerous problems that fall into this class of connected clustering problems, such as area aggregation in map generalization~\cite{oehrlein17} and political districting~\cite{ValidiB22}. Almost all of these problems are NP-hard, and solving them to optimality has proven to be a difficult task. Integer Linear Programming (ILP) techniques have shown great success for many difficult combinatorial optimization problems, and thus the majority of exact algorithms in the area of connected clustering problems are based on ILP formulations.

One instance of an NP-hard connected clustering problem is the $p$-regions problem, introduced by Duque et al.~\cite{duque11}. They consider the problem of clustering a set of given areas into a predefined number of $p$ spatially contiguous regions, while minimizing the total heterogeneity of all regions. The authors define the heterogeneity of a region by the sum of dissimilarities between all pairs of vertices within a region, where the dissimilarity function $d(u,v)$ for all pairs of areas $u,v$ is given as part of the problem input. Figure \ref{fig:germany-part} shows solutions to the $p$-regions problem for unemployment data of Germany, with $k\in\{3,6\}$.

In their work, the authors frame the $p$-regions problem as a graph problem and propose the three ILP models \texttt{Tree}, \texttt{Order} and \texttt{Flow}~\cite{duque11}. Their computational experiments show that even small instances (\# areas $\approx 25$) are already challenging and could not be solved within the time limit of 3 hours.

The $p$-regions problem is closely related to both the $k$-partitioning problem and connected clustering problems.
In the $k$-partitioning problem, the goal is to partition the vertices of an edge-weighted graph into $k$ subsets such that the total weight of edges whose endpoints lie in the same part is minimized. It is equivalent to the max-$k$-cut problem, where the goal is to maximize the weight of edges running between partitions. Notably, Chopra et al.~\cite{chopra93,chopra95}, Ales et al.~\cite{ales16} and Fairbrother et al.~\cite{Fairbrother17} studied the problem from a polyhedral perspective and identified various facet-defining inequalities. They consider formulations with only edge variables, edge and representative variables, and edge and vertex assignment variables.

Recently, there have been a number of ILP approaches for the connected max-$k$-cut problem, which is defined as the max-$k$-cut problem with the additional constraint that each partition must be connected~\cite{hojny21,healy24}. Notably, Healy et al.~\cite{healy24} proposed an ILP model based on edge cut constraints. In contrast to the $p$-regions problem, where the costs are defined over every pair of vertices, the objective is defined over the cost of cut edges. 
In the p-regions problem, assigning two non-adjacent vertices to different parts is associated with a cost, whereas the connected max-k-cut problem does not penalize such assignments. Therefore the $p$-regions problem is more general.

Another closely related problem is the \textit{graph-connected clique-partitioning problem (GCCP)}, introduced in~\cite{benati17}. As in the $p$-regions problem, the cost function is defined over each vertex pair and the objective is to minimize the sum of costs over all vertex pairs in the same partition. The difference between the two problems is that there is no constraint on the number of partitions for GCCP, while for the $p$-regions problem the number of regions is fixed. In~\cite{benati17}, Benati et al.~introduce three ILP formulations, based on a classic flow formulation, a formulation based on Miller-Tucker-Zemlin constraints and a cut-based formulation. In a subsequent paper, they propose a branch-and-price approach for this problem~\cite{benati22}.

\subsection*{Our Contribution}
In this work, we present and theoretically analyze new ILP formulations for the $p$-regions problem. Based on these formulations, we develop a branch-and-cut algorithm and evaluate its performance experimentally in a computational study, demonstrating significant improvements over the state of the art.
In particular:
\begin{itemize}
    \item We propose a new ILP formulation \ERS for the $p$-regions problem exploiting a connection to the $k$-partitioning problem. Specifically, we combine the edge-representative model for the $k$-partitioning problem~\cite{ales16} and the vertex-separator constraints for modeling graph connectivity~\cite{oehrlein17}.
    \item We present constraints to strengthen \ERS and \tree~\cite{duque11} as well as a combination of the strengthened models into a new hybrid model \ERSTree.
     \item We conduct a polyhedral analysis of our new models and known formulations from the literature~\cite{duque11}, 
     \tree and \flow, showing that our new hybrid model \ERSTree strictly dominates.
     \item We provide implementations of our branch-and-cut algorithms based on the formulations \ERSTree, \ERS, \tree, and \flow, which are publicly available \cite{faber26}.
    \item We experimentally evaluate the performance of our new algorithms on a benchmark set of real-world instances as well as generated grid instances for multiple numbers of regions.  \ERS and \ERSTree display superior performance across the benchmark set, supporting our theoretical findings. In particular, neither \flow nor \tree was able to solve the instance \texttt{Romania} for any tested value for $k$, while \ERSTree manages to solve it for all but one of the eight tested values for $k$.
\end{itemize}

The rest of this paper is organized as follows: In section \ref{sec:preliminaries}, we introduce the notation used in this paper and formally define the $p$-regions problem. In section \ref{sec:sota}, we revisit state-of-the-art ILP formulations for the $p$-regions problem and related problems that our formulations build upon. In section \ref{sec:treeplus}, we present a strengthened version of the model \tree proposed in \cite{duque11}.
Section \ref{sec:ers}, introduces our new models \ERS and \ERSTree. The relaxation strength of the models is analyzed in section \ref{sec:theory}.
In section \ref{sec:separation}, we describe the separation routines of our branch-and-cut algorithms. The results of our experimental evaluation are discussed in section \ref{sec:experiment}. Finally, \ref{sec:conclusion} concludes the paper.
\section{Preliminaries}\label{sec:preliminaries}
We use the common notation ${V \choose i}$ to describe the set of all subsets of $V$ that have cardinality $i$. We also use the notations $[b]$ and $[a,b]$ for $b \in \mathbb{N}_0$ to denote the set $\{1,\ldots,b\}$ and $\{a,\ldots,b\}$ of natural numbers.

In this paper, we mainly consider simple undirected graphs $G=(V(G),E(G))$ with dissimilarities $d\colon {V \choose 2} \rightarrow \mathbb{R}$ defined between any pair of vertices. Following standard notation, we define $n = |V(G)|, m = |E(G)|$. We consider the set of vertices to be numbered, i.e. $V = [n]$. We will use $V$ and $E$ instead of $V(G)$ and $E(G)$ if $G$ is clear from the context.

For a vertex set $V' \subseteq V$, $G[V']$ denotes the subgraph induced by $V'$, with vertex set $V'$ and edge set $E(V') = \{e \in E\mid e \subseteq V'\}$.
For an undirected graph $G=(V,E)$, we define the neighborhood $N(v) = \{w \in V(G) \mid \{v,w\} \in E(G)\}$
For a directed graph $D=(V,A)$, we define $N^-(v) = \{u \in V(D) \mid (u,v) \in A(D)\}$ and $N^+(v) = \{w \in V(D) \mid (v,w) \in A(D)\}$.

Since some algorithms and formulations are defined only for directed graphs, we associate with a set of undirected edges $E$ the set of directed edges $\diedge{E} \subseteq V \times V$ containing both $(u,v)$ and $(v,u)$ for every undirected edge $\{u,v\} \in E$.

We now formally define the $p$-regions problem. We remark that in the following, we denote the number of regions requested as output by $k$ (as opposed to $p$), as this is standard convention in the graph partitioning literature ~\cite{ales16,healy24,chopra93,Fairbrother17}. 
\begin{definition}[$p$-regions problem]
    Given a connected graph $G=(V,E)$, pairwise vertex dissimilarities $d\colon {V \choose 2} \rightarrow \mathbb{R}_{\geq 0}$ and an integer $k \in [n]$, the $p$-regions problem is to find a $k$-partitioning $\mathcal{P} = (V_1,...,V_k)$ of $V$ such that
    \begin{enumerate}
        \item $V = \bigcup_{i \in [k]} V_i$ (each vertex is covered)
        \item $V_i \cap V_j = \emptyset$ for $i \neq j$ (the partitions are disjoint)
        \item $G[V_i]$ for $i \in [k]$ is connected
        \item $V_i \neq \emptyset$ for all $i \in [k]$
    \end{enumerate}
    minimizing the total heterogeneity $H(\mathcal{P}) = \sum_{i \in [k]}\sum_{u,v \in V_i} d(u,v)$.
\end{definition}

\section{State-of-the-art ILP formulations}\label{sec:sota}
\subsection{ILP formulations for the $p$-regions problem}
In this section we review the ILP models for the $p$-regions problem proposed by Duque et al.~\cite{duque11}. In this work, the authors proposed three ILP models, namely \texttt{Flow}, \texttt{Tree} and \texttt{Order}. We will discuss the first two models, as \texttt{Order} has been shown to not be competitive~\cite{duque11}. 

\subsubsection{Formulation \texttt{Flow}}
\begin{model}[ht]
\begin{align}
 &\text{min}  && \sum_{u < v}d(u,v)\cdot x_{u,v}\notag   &&\\
&  \text{s.t.} && x_{u,v}\geq z_{u,i} + z_{v,i} - 1   &&  \forall u,v \in V,u \neq v,i \in [k]\label{eq:flow2} \\
 &&& \sum_{i\in [k]} z_{v,i} = 1   &&  \forall v \in V\label{eq:flow3}  \\
&&& \sum_{u\in V} r_{u,i} = 1  & \forall &i\in [k] \label{eq:flow4}\\ 
&&& r_{u,i} \leq z_{u,i} & \forall &u\in V,i \in [k] \label{eq:flow5}\\
&&& f^i_{u,v} \leq (M-1) \cdot z_{u,i}& \forall &(u,v)\in \diedge{E},i \in [k] \label{eq:flow6}\\
&&& f^i_{u,v} \leq (M-1) \cdot z_{v,i}& \forall &(u,v)\in \diedge{E},i \in [k] \label{eq:flow7}\\
&&& \sum_{u \in N(v)}f^i_{u,v} - \sum_{w \in N(v)}f^i_{v,w} \geq z_{v,i} - M\cdot r_{v,i} & \forall &v\in V,i\in[k] \label{eq:flow8}\\
&&& z_{1,1} = 1 \label{eq:flow9}\\
&&&             x_{u,v} = x_{v,u} &&  \forall (u,v) \in V\times V,u\neq v\label{eq:flow13}   \\
&&&             z_{v,i},r_{v,i} \in \{0,1\}  &&  \forall v \in V,i \in [k]\label{eq:flow10}   \\
&&&             x_{u,v} \in \{0,1\}  &&  \forall (u,v) \in V \times V,u \neq v \label{eq:flow11}\\
&&&f^i_{u,v}\in [0,M-1] &\forall& (u,v)\in \diedge{E}, i \in [k]\label{eq:flow12}
\end{align}
\caption{Model \flow proposed in~\cite{duque11}} \label{model:flow}
\end{model}

Formulation \texttt{Flow}~\cite{duque11} (see Model \ref{model:flow}) contains three sets of binary and one set of continuous variables. The first set consists of binary variables $z_{v,i}$ for each vertex $v \in V$ and region $i \in [k]$, where $z_{v,i} = 1$ iff $v \in V_i$. Additionally, the model has binary variables $x_{u,v}$ for each vertex pair $u,v \in V$ with $u \neq v$, defined by $x_{u,v} = 1$ exactly if $u$ and $v$ are in the same region (denoted as \textit{pairing variables}). To enforce connectivity within a region, the model uses a standard multi-commodity flow formulation, using flow variables $f^i_{u,v}$ for $(u,v) \in \diedge{E}, i \in [k]$ and binary root variables $r_{u,i}$ for $u \in V, i \in [k]$. 

\begin{model}[ht]
\begin{align}  
& \text{min}  && \sum_{u<v}d(u,v)\cdot x_{u,v}\notag    &&\\
& \text{s.t.} && \sum_{(u,v) \in \diedge{E}}y_{u,v} = n - k \label{eq:tree1}  \\
&&& \sum_{v \in N(u)}y_{u,v} \leq 1  &&  \forall u\in V\label{eq:tree2}  \\
&&& x_{u,w} \geq x_{u,v}+x_{v,w}-1  &&  \forall u,v,w \in {V \choose 3}\label{eq:tree3}  \\
&&& y_{u,v} \leq x_{u,v} &&  \forall (u,v) \in \diedge{E}  \label{eq:tree4} \\
&&& \sum_{(u,v) \in \diedge{S}}y_{u,v} \leq |S|-1 && \forall S \subseteq V \label{eq:tree5} \\
&&& x_{u,v} = x_{v,u}&& \forall (u,v) \in V \times V, u \neq v\\
&             && y_{u,v} \in \{0,1\}  &&  \forall (u,v) \in \diedge{E}\label{eq:tree6}   \\
&             && x_{u,v} \in \{0,1\}  &&  \forall (u,v) \in V \times V, u \neq v\label{eq:tree7}  
\end{align}
\caption{Model \tree proposed in~\cite{duque11}}\label{model:tree}
\label{model-tree}
\end{model}

The objective minimizes the total sum of dissimilarities between vertices within the same region in the partitioning. Constraints \eqref{eq:flow2} ensure that $x_{u,v} = 1$ if there exists an $i \in [k]$ such that $z_{u,i} = z_{v,i} = 1$. Finally, constraints \eqref{eq:flow3} enforce that each vertex is assigned to exactly one region. Constraints \eqref{eq:flow4} -- \eqref{eq:flow8} model regional connectivity via a flow network for each region. Constraints \eqref{eq:flow4} -- \eqref{eq:flow5} ensure that for each partition, exactly one vertex is assigned as the root vertex. Constraints \eqref{eq:flow6} -- \eqref{eq:flow7} ensure that flow of type $i \in [k]$ can only traverse an edge if both incident vertices are assigned to $i$. Here, $M$ denotes a sufficiently large constant. The authors use $M = n-k+1$, as a partition can have at most $n-k+1$ vertices. Finally, the flow conservation constraints \eqref{eq:flow8} enforce that each vertex assigned to partition $i$ that is not the root must consume exactly one unit of flow. Constraint \eqref{eq:flow9} has the purpose of breaking some symmetries arising from permuting the region indices.

\subsubsection{Formulation \texttt{Tree}}
The tree formulation~\cite{duque11} (see Model \ref{model:tree}) models each region as a directed arborescence (oriented towards the root) in the graph and is defined over the directed edge set $\diedge{E}$.  Like formulation \flow, the model contains pairing variables $x_{u,v}$. Additionally, it uses linking variables $y_{u,v}$ for all $(u,v) \in \diedge{E}$, where $y_{u,v} = 1$ exactly if $(u,v)$ is part of an arborescence. To enforce that the selected edges form $k$ trees spanning all vertices, one can fix the number of selected edges to $n-k$ and ensure that the set of selected edges do not contain any cycle. The authors propose two variants of cycle-breaking constraints, a compact encoding and a formulation with up to exponentially many cycle breaking constraints. We will only consider the latter one, as the compact formulation was shown to be non-competitive.

The objective is equivalent to the one in \flow. Constraint \eqref{eq:tree1} enforces that exactly $n-k$ edges are selected, the total number of edges contained in $k$ edge-disjoint trees spanning all vertices. Constraints \eqref{eq:tree2} enforce that each vertex can have at most one outgoing arc, as the arborescences are oriented towards the root.  Constraints \eqref{eq:tree3} are the so-called \textit{transitivity constraints}: They ensure that for any triple of vertices $u,v,w$, if $u,v$ and $v,w$ are in the same region, then $u$ and $w$ are also in the same region, yielding a consistent pairing. Constraints \eqref{eq:tree4} state that an edge can only be selected if the incident vertices are in the same region. Finally, the subtour elimination constraints \eqref{eq:tree5} ensure that for any set $S \subseteq V$ the set of selected edges never form a cycle in $G[S]$ by imposing
that $G[S]$ contains at most $|S|-1$ selected edges. 

\subsection{ILP formulation for $k$-partitioning}
The $k$-partitioning problem, as defined by Ales et al.~\cite{ales16}, is as follows: Given an edge-weighted complete graph $H = (V(H),E(H))$ with weights $d\colon E(H)\rightarrow \mathbb{R}$, find a partitioning of the vertices into $k$ parts minimizing the total weight of edges whose endpoints lie in the same part.
Ales et al.~\cite{ales16} proposed two ILP formulations, the \textit{node-cluster} formulation and the \textit{edge-representative} formulation. As the node-cluster formulation suffers from a poor LP relaxation and symmetries, we present the edge-representative formulation \ER (see Model~\ref{model:er}). 

\begin{model}[bht]
\begin{align}  
&  \text{min} && \sum_{u<v}d(u,v)\cdot x_{u,v} \notag   &&\\
& \text{s.t.} &&  x_{u,w} \geq x_{u,v}+x_{v,w}-1  &&  \forall u,v,w \in {V \choose 3}\label{eq:ERS2}  \\
&&& r_v \leq 1 - x_{u,v} &&  \forall u,v \in V, u < v \label{eq:ERS3}  \\
&&& r_v \geq 1 - \sum_{u < v} x_{u,v} &&  \forall v \in V\label{eq:ERS4}  \\
&&& \sum_{v \in V} r_v = k &&  \label{eq:ERS5}  \\
&&& x_{u,v} = x_{v,u}&&  \forall (u,v) \in V \times V, u \neq v\\
&             && r_{v} \in \{0,1\}  &&  \forall v \in V\label{eq:ERS6}   \\
&             && x_{u,v} \in \{0,1\}  &&  \forall (u,v) \in V \times V, u \neq v\label{eq:ERS7}  
\end{align}
\caption{Model \ER proposed in~\cite{ales16}}\label{model:er}
\end{model}

Like \tree, the formulation contains pairing variables $x_{u,v}$ and transitivity constraints \eqref{eq:ERS2}. Additionally, the formulation contains binary variables $r_v$ for each $v \in V$, which we denote as \textit{representative variables}. 

The underlying idea is that we designate exactly one vertex per part as a representative, so that we can constrain the number of parts by constraining the number of representatives. By convention, a part is represented by the vertex in the part having the lowest index.
By constraints \eqref{eq:ERS3} and \eqref{eq:ERS4}, it holds that $r_v = 1$ exactly if $v$ is paired with no other vertex $u$ with $u < v$ (in which case we call $v$ a representative): \eqref{eq:ERS3} enforces that a vertex $v$ cannot be a representative if it is paired with a vertex $u < v$, while \eqref{eq:ERS4} enforces that if no vertex with smaller index is paired with $v$, $v$ must be a representative.  Therefore, there must be exactly one representative per region, so the number of regions can be restricted by fixing $\sum_{v \in V} r_v = k $ in constraint \eqref{eq:ERS5}.

\subsection{ILP formulations for spatial contiguity}
Oehrlein and Haunert~\cite{oehrlein17} consider a similar districting problem, which also requires a partitioning of the adjacency graph $G$ into connected regions. In contrast to our problem, they use an objective based on region centers instead of our pairwise objective and do not prescribe the number of output regions. Nevertheless, the connectivity constraints used in their model can still be applied to our setting. The constraints are based on the concept of vertex separators, defined as follows:
\begin{definition}[$u,v$-Separator]
    Given a graph $G = (V,E)$ and vertices $u,v \in V$, a $u,v$-Separator $S\subseteq V$ is a set of vertices such that $G[V\setminus S]$ does not contain a path between $u$ and $v$.
\end{definition}
A graph is connected exactly if for any $u,v \in V$, every \uvsep is non-empty. Therefore, connectivity can be expressed using the following constraints: 
\begin{align}\label{eq:separator}
x_{u,v} \leq \sum_{w \in S} x_{u,w}  \quad\quad \forall  u,v \in V,u\neq v,S \in \mathcal{S}_{u,v}
\end{align}
Here, $\mathcal{S}_{u,v}$ is the set of all $u,v$-separators, while variables $x_{u,v} = 1$ exactly if vertices $u$ and $v$ are paired together. 

\section{\treeplus: Strengthening Formulation \tree} \label{sec:treeplus}
Before introducing our new algorithm \ERS in the next section, we propose two modifications of \tree presented in~\cite{duque11,duque18}. 

The first modification is strengthening constraints \eqref{eq:tree4} by substituting it with
\begin{align}
    &&& y_{u,v}+y_{v,u} \leq x_{u,v} &&  \forall \{u,v\} \in E  \label{eq:tree8}
\end{align}
The idea is that the directed edges $(u,v)$ and $(v,u)$ cannot be part of the same arborescence. The resulting constraints are at least as strong as \eqref{eq:tree4} as $y_{u,v} \leq y_{u,v}+y_{v,u}$.

The second modification is a new type of valid inequality based on the cycle-breaking constraints \eqref{eq:tree5}:

\begin{equation}
    2\cdot\sum_{(u,v) \in \diedge{T}}y_{u,v} - \sum_{\{u,v\} \in T}x_{u,v}\leq |S|-2 \quad\quad \forall S \subseteq V, T\subseteq E(S) : |T| \geq |S| \label{eq:tree-strong} 
\end{equation}

\begin{theorem}
    Every valid integer solution $(x^*,y^*)$ of formulation \tree satisfies constraints \eqref{eq:tree-strong}.
\end{theorem}
\begin{proof}
    Let $S \subseteq V$ and $ T\subseteq E(S)$ with $|T| \geq |S|$ and let $F = \{a \in \diedge{T} \mid y^*_a = 1 \}$. As $F$ induces a forest in $G$ and by extension $G[S]$, it holds that $|F| \leq |S|-1$. We distinguish between two cases:
    \begin{description}
    \item[Case 1: $|F| \leq |S|-2$:] By constraints \eqref{eq:tree4}, we have 
    \begin{equation*}
        \sum_{(u,v) \in \diedge{T}}y^*_{u,v} - \sum_{\{u,v\} \in T}x^*_{u,v} \leq 0
    \end{equation*}
    from which constraint \eqref{eq:tree-strong} follows.
    \item[Case 2: $|F| = |S|-1$:] Because, by construction, the edges in $F$ cannot contain a cycle, they must form an arborescence spanning all vertices in $S$. Therefore by constraints \eqref{eq:tree4} and the transitivity constraints \eqref{eq:tree3}, all vertices in $S$ are paired together, which implies that $x^*_{u,v} = 1$ for all $\{u,v\} \in T$.
    As $|T| \geq |S|$, it must hold that
      $  \sum_{\{u,v\} \in T}x^*_{u,v} = |T| \geq |S| $
    which completes the proof, as it must hold that $2 \cdot \sum_{(u,v) \in \diedge{T}}y^*_{u,v} \leq 2(|S|-1)$.
\end{description}
\end{proof}

Furthermore, we can show that for sets $S$ with $|E(S)| = |S|$ (e.g. chordless cycles), constraints \eqref{eq:tree-strong} dominate constraints \eqref{eq:tree5}.
\begin{observation}\label{obs:tree-strong}
    For vertex sets $S$ with $|E(S)| = |S|$, constraints \eqref{eq:tree-strong} imply constraint \eqref{eq:tree5}.
\end{observation}
For the proof, see Section~\ref{proof:tree-strong} of the appendix.


The resulting improved model \treeplus is defined as Model \ref{model:treeplus}.
\begin{model}
\begin{align}  
&  \text{min} && \sum_{u<v}d(u,v)\cdot x_{u,v}  \notag  &&\\
& \text{s.t.} && \eqref{eq:tree1}-\eqref{eq:tree3},\eqref{eq:tree5}-\eqref{eq:tree7},\eqref{eq:tree8},\eqref{eq:tree-strong}\notag
\end{align}
\caption{Model \treeplus}\label{model:treeplus}
\end{model}

\section{New ILP formulations \ERS and \ERSTree} \label{sec:ers}
\subsection{\ERSBase}
One can observe that the $p$-regions problem essentially is a $k$-partitioning problem under the constraint that every region must be connected. It is important to note that, in this setting, the complete graph $H=(V,E(H))$ defining the objective and the adjacency graph $G = (V,E(G))$ defining the connectivity are independent from each other. Given the adjacency graph $G$ and dissimilarities $d: {V \choose 2} \rightarrow \mathbb{R}$, define the dissimilarity graph to be the complete graph $H=(V,E(H))$ with edge weights defined by $d$. The $p$-regions problem can then be stated as finding an optimal $k$-partitioning in $H$, with the additional constraint that each partition $V_i$ must induce a connected subgraph $G[V_i]$ in $G$. 

Following this motivation, we combine ideas from the presented state-of-the art ILP formulations for $k$-partitioning with formulations for enforcing regional connectivity.  
We propose a new ILP model for the $p$-regions problem, which is derived by combining the edge-representative model for $k$-partitioning~\cite{ales16} and the vertex separator constraints for connectivity~\cite{oehrlein17}. As in the edge-representatives model for $k$-partitioning, we define binary variables $x_{u,v}$ (called \emph{pairing variables}) for each $u,v \in V$ with $x_{u,v} = 1$ iff $u$ and $v$ share a partition and binary representatives variables $r_v$ for $v \in V$, with $r_v = 1$ iff $v$ is a representative. The model $\ERSBase$ (\textbf{e}dge-\textbf{r}epresentatives model with \textbf{s}eparator constraints) is defined as Model \ref{model:ersbase}.
\begin{model}[ht]
\begin{align}  
&  \text{min}  && \sum_{u<v}d(u,v)\cdot x_{u,v} \notag  &&\\
& \text{s.t.} && \eqref{eq:ERS2}-\eqref{eq:separator}\notag
\end{align}
\caption{Model \ERSBase }\label{model:ersbase}
\end{model}
The model contains $\mathcal{O}(n^2)$ variables, $\mathcal{O}(n^3)$ constraints of type \eqref{eq:ERS2}-\eqref{eq:ERS7} and up to $\mathcal{O}(2^nn^2)$ constraints of type \eqref{eq:separator}. The correctness of the model directly follows from the correctness of the edge-representative model for the $k$-partitioning problem and the correctness of the vertex separator constraints. This model forms the basis of the following models \ERS and \ERSTree. Notably, it does not contain any symmetries like \flow, so we do not need to consider any symmetry breaking constraints.
\subsection{$\ERS$}
The model $\ERSBase$ already correctly defines the feasible solutions and objective of the $p$-regions problem. To strengthen the LP relaxation of the model, we additionally include further valid constraints.

\medskip\noindent
\textbf{$k$-Forest Constraint.} The $k$-forest constraint follows from the observation that any connected partition contains a spanning tree with exactly $|V_k|-1$ edges over its vertices.

As the incident vertices of these trees are paired together, the following inequality is valid:
\begin{equation}
     \sum_{\{u,v\} \in E}x_{u,v} \geq n-k \label{eq:edge-sum}
\end{equation}
\textbf{General Clique Constraints.}
General clique constraints were introduced in Chopra et al.~\cite{chopra93} for the $k$-partitioning problem and were proven to be facet defining for the $k$-partitioning polytope. They are based on the idea that in a set $Q$ of $|Q| = qk+p$ vertices (with $q \in [\lfloor\frac{n}{k}\rfloor], p \in [0,k-1]$), minimizing the number of vertex pairs in $Q$ that are in the same region is achieved by spreading the vertices of $Q$ into $p$ clusters of size $q+1$ and $k-p$ clusters of size $q$. The resulting lower bound on the number of paired vertices yields the following valid constraints:
\begin{align}\label{eq:gen-clique}
    \sum_{u,v \in Q}x_{u,v} \geq {q+1 \choose 2}p+{q \choose 2}(k-p)  && \forall Q \subseteq V,|Q| = qk+p \geq k+1
\end{align}

The constraints corresponding to subsets $Q$ with $|Q| = k+1$ are also simply known as clique inequalities~\cite{chopra93}.

\medskip\noindent
\textbf{Formulation $\ERS$.}
We define the model \ERS (see Model \ref{model:ers}) to be the model $\ERSBase$ strengthened with the general clique constraints and the $k$-forest constraint.
\begin{model}
\begin{align}  
& \text{min}  && \sum_{u<v}d(u,v)\cdot x_{u,v}\notag   &&\\
& \text{s.t.} && \eqref{eq:ERS2}-\eqref{eq:separator},\eqref{eq:edge-sum},\eqref{eq:gen-clique}\notag
\end{align}
\caption{Model \ERS}\label{model:ers}
\end{model}
The model contains the same set of variables as $\ERSBase$. It has one additional constraint \eqref{eq:edge-sum} and $\mathcal{O}(2^n)$ additional constraints of type \eqref{eq:gen-clique}.
\subsection{\ERSTree}
To further tighten the model, we consider a combination of the models \ERS and \treeplus, denoted as \ERSTree (see Model \ref{model:erstree}). 
\begin{model}
\begin{align}  
& \text{min}  && \sum_{u<v}d(u,v)\cdot x_{u,v} \notag   &&\\
& \text{s.t.}  && \eqref{eq:ERS2}-\eqref{eq:separator},\eqref{eq:edge-sum},\eqref{eq:gen-clique}\notag\\
&&&\eqref{eq:tree1},\eqref{eq:tree3},\eqref{eq:tree5}-\eqref{eq:tree7},\eqref{eq:tree8},\eqref{eq:tree-strong}\notag\\
&&& \sum_{v \in N(u)}y_{u,v} = 1 - r_u &&  \forall u\in V\label{eq:ers-tree1}  
\end{align}
\caption{Model \ERSTree}\label{model:erstree}
\end{model}
Constraints \eqref{eq:ers-tree1} are a strengthened version of constraints \eqref{eq:tree2}, which exploit the presence of root variables in formulation \ERS. We consider the arborescenes to be rooted in the root vertices defined by variables $r$ in \ERS. A vertex has exactly one outgoing arc in the arborescence if it is not a root and otherwise zero. This has the additional benefit of eliminating some symmetries that arise from the ambiguous choice of trees for each partition. 

The polyhedral analysis in the next section demonstrates that combining the constraints of \ERS and \treeplus into \ERSTree yields a model that can be strictly stronger than both.

\section{Theoretical Analysis}\label{sec:theory}
In this section, we conduct a theoretical analysis of the different ILP formulations, proving that \ERSTree has a strong relaxation.
We compare the relaxation strength of the SOTA formulations \tree and \flow~\cite{duque11} with our new formulations \ERS and \ERSTree. We will use the following notation: $P(\square) \subseteq \mathbb{R}^{|\square|}$ describes the set of feasible solutions for the LP relaxation of $\square \in \{\ERS,\ERSTree,\tree,\flow\}$ (where $|\square|$ is the number of variables in $\square$). Furthermore, for a model $\square$ and a set $\Delta$ of variables contained in $\square$, $P_\Delta(\square)$ describes the projection of $P(\square)$ onto the space of variables $\Delta$.

In the following, we will compare the polytopes of the models with respect to their projection onto the pairing variables $P_x(\square)$ for the following three reasons: (1) These variables are the only variables relevant for the objective, (2) they are included in every model, (3) they completely describe a partitioning without any arbitrary symmetries (e.g. region indices).
We assume that $M$ is chosen to be $n-k+1$ for \flow, as every lower value for $M$ leads to an incomplete model (i.e. the model excludes feasible solutions).

We first remark that the relaxation of \flow is very weak, which can be seen from the following theorem:
\begin{theorem}\label{theorem:flow-relax}
For $k \geq 3$, $P_x(\ERS) \subseteq P_x(\flow)$ and $P_x(\tree) \subseteq P_x(\flow)$ and both inclusions can be strict. In particular, for any instance with $3 \leq k \leq n$, it holds that $P_x(\flow) = [0,1]^{(n^2-n)}$ and for $k < n$ the inclusion is always strict.
\end{theorem}
\begin{proof}(Sketch)
    In a fractional solution, we can assign every vertex to be partially assigned to every region $i \neq 1$ by $z_{v,i} = \frac{1}{k-1}$. This results in the right-hand side of constraints \eqref{eq:flow2} to be zero and the pairing variables to be unbounded. To show that these conclusion is strict for any $k < n$, we can observe that neither \ERS nor \tree contain the solution $x_{u,v} = 0$ for all $u,v \in V$. For the full proof, see Section \ref{proof:flow-relax} of the appendix. 
\end{proof}

\begin{theorem}\label{theorem:treestrong}
     $P_x(\treeplus) \subseteq P_x(\tree)$ and this inclusion can be strict.
\end{theorem}
\begin{proof}(Sketch)
    The inclusion trivially holds. To see that the inclusion can be strict, we can consider a triangle, and a fractional assignment of $\frac{2}{3}$ for every pairing and every arc in a directed cycle in the triangle. We observe that this satisfies constraints \eqref{eq:tree5} but violates constraints \eqref{eq:tree-strong}. For the complete proof, see Sections \ref{proof:treestrong} of the appendix.
\end{proof}

\begin{theorem}\label{theorem:erstree-incomp}
     $P_x(\ERS) \not\subset P_x(\tree)$ and $P_x(\tree) \not \subset P_x(\ERS)$, i.e. the two models are incomparable under inclusion. The same holds for $P_x(\ERS)$ and $P_x(\treeplus)$.
\end{theorem}
\begin{proof}(Sketch)
    To show that $P_x(\tree) \not \subset P_x(\ERS)$, we can consider the instance depicted in Figure \ref{fig:seperator-implied-counter} with a corresponding solution of \treeplus for $k=2$. Notably, it violates the vertex separator constraints of 
    \ERS. 
    
    For $P_x(\ERS) \not \subset P_x(\tree)$, consider the example depicted in Figure \ref{fig:cycle-implied-counter} and the corresponding solution for \ERS for $k=4$: To satisfy constraint \eqref{eq:tree1}, we would need to violate the subtour elimination constraint for $S = \{v_5,v_6,v_7,v_8\}$. 

    For the full proof, see Section \ref{proof:erstree-incomp} of the appendix.
\end{proof}
\begin{figure}
    \centering
    \includegraphics[]{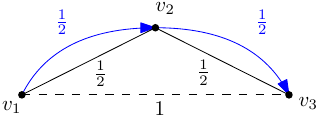}
    \caption{Example for $P(\treeplus) \not \subset P(\ERS)$, $k=2$: the values of $y$ are in blue, the values for $x$ in black. Solid lines depict existing edges.}
    \label{fig:seperator-implied-counter}
\end{figure}

\begin{figure}
    \centering
    \includegraphics[]{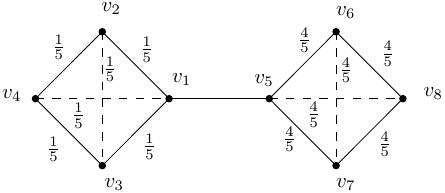}
    \caption{Example for $P(\ERS) \not \subset P(\tree)$, $k=4$. The figure depicts the non-zero values of $x$, existing edges are drawn as solid lines.  The representative variables have values: $r_1,r_5 = 1,r_2=\frac{4}{5},r_3 = \frac{3}{5},r_4 =\frac{2}{5},r_6=\frac{1}{5},r_7,r_8 =0$}
    \label{fig:cycle-implied-counter}
\end{figure}

\begin{theorem}\label{theorem:ersstrong}
     $P_x(\ERSTree) \subseteq P_x(\treeplus)$ and $P_x(\ERSTree) \subseteq P_x(\ERS)$ and both inclusions can be strict.
\end{theorem}
\begin{proof}
    The inclusion $P_x(\ERSTree) \subseteq P_x(\ERS)$ trivially follows from \ERSTree containing all constraints of \ERS. Similarly, $P_x(\ERSTree) \subseteq P_x(\treeplus)$  follows from \ERSTree containing all constraints of \treeplus, with the exception of substituting constraints \eqref{eq:tree2} by \eqref{eq:ers-tree1}. However, \eqref{eq:tree2} is implied by \eqref{eq:ers-tree1} as for the right-hand side of the inequalities it holds that $1-r_v \leq 1$. The strict inclusions follow from Theorem \ref{theorem:erstree-incomp}.
\end{proof}
In summary, we have the inclusion hierarchy as depicted in Figure \ref{fig:inc-hierarchy}.
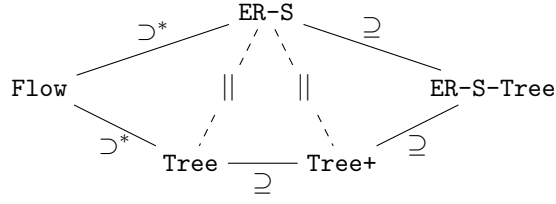
\begin{figure}
    \centering
\begin{tikzpicture}[node distance=0cm]

\node (A) {\flow};
\node (B) [above right of=A,xshift=3cm,yshift=1cm] {\ERS};
\node (C) [below right of=A,xshift=2cm,yshift=-1cm] {\tree};
\node (Cp) [below right of=A,xshift=4cm,yshift=-1cm] {\treeplus};
\node (D) [right of=A, xshift=6cm] {\ERSTree};

\draw (A) -- node[above] {$\supset^*$} (B);
\draw (A) -- node[below] {$\supset^*$} (C);
\draw (B) -- node[above] {$\supseteq$} (D);
\draw (C) -- node[below] {$\supseteq$} (Cp);
\draw (Cp) -- node[below] {$\supseteq$} (D);

\draw[dashed] (B) -- node[midway,fill=white]{$||$}(C);
\draw[dashed] (B) -- node[midway,fill=white]{$||$}(Cp);
\end{tikzpicture}
    \caption{Inclusion hierarchy for the projected polytopes of the five analyzed models. The strict inclusion $\supset^*$ only holds for $3 \leq k \leq n-1$}
    \label{fig:inc-hierarchy}
\end{figure}

\FloatBarrier

\section{Separation algorithms} \label{sec:separation}
With the exception of \flow, all presented formulations contain a potentially exponential number of constraints. These require a branch-and-cut implementation, where we start by solving a relaxed version of the models omitting the constraint classes having exponential size and iteratively add them. 

\medskip\noindent
\textbf{Subtour elimination constraints.} Both \tree and \treeplus contain the subtour elimination constraints \eqref{eq:tree5}. For separating integer solutions, we follow the approach by Duque et al.~\cite{duque18}, which is to compute the set of strongly connected components induced by all arcs $(u,v)$ with $y_{u,v} = 1$ in the current solution. Constraint \eqref{eq:tree2} guarantees that there is a violated subtour elimination constraint exactly if there exists a strongly connected component of size at least~2. 

For separating fractional solutions of the LP relaxation, we employ the classical algorithm by Magnanti and Wolsey~\cite{magnanti1995} that is based on reducing the problem to a min-cut problem on an auxiliary graph.

To separate our new constraints \eqref{eq:tree-strong}, we use the following heuristic:
using the above algorithm, we start with a set $S$ minimizing $|S|- \sum_{(u,v) \in \diedge{S}}y_{u,v}$. We then compute $T = \{\{u,v\} \in E(S) \mid 2 (y_{u,v}+y_{v,u}) - x_{u,v} > 0\}$. If constraint \eqref{eq:tree-strong} is violated for $S$ and $T$, we add it to our model.
We note that the algorithm by Magnanti and Wolsey cannot be directly applied to separate \eqref{eq:tree-strong}, as the associated costs $2 (y_{u,v}+y_{v,u}) - x_{u,v}$ for an edge $\{u,v\}$ can be negative, prohibiting the reduction to the min-cut problem. We do not know if the separation problem is polynomial-time solvable or NP-hard.

\medskip\noindent
\textbf{Vertex separator constraints.} The separation problem of constraints \eqref{eq:separator} for fractional solutions can be solved in polynomial time by transforming the problem of finding a minimum $u,v$-separator into a minimum $u,v$-cut problem. This can be done via a standard vertex-splitting graph transformation, see Oehrlein and Haunert~\cite{oehrlein17}. For each pair of vertices, we run the algorithm of finding a minimum $u,v$-separator twice, once with costs $x_{u,a}$ for $a \in V \setminus\{u,v\}$ and once with costs $x_{v,a}$ for $a \in V \setminus\{u,v\}$. This results in $\mathcal{O}(n^2)$ calls to a flow algorithm on a graph of size $\mathcal{O}(n+m)$. Additionally, we use \textit{nested cuts}~\cite{KochM98}, which is a technique to find multiple orthogonal cuts in one iteration.

\medskip\noindent
\textbf{General clique constraints.} The separation of general clique constraints is NP-hard, as the independent set problem can be reduced to it. Therefore, we use a greedy heuristic incorporating local search steps to separate constraints \eqref{eq:gen-clique}. In each iteration, the algorithm greedily picks the best vertex to extend the current set $S$. It then checks if it can improve the current set by swapping out a vertex $b \in S$ with a vertex $a\in V\setminus S$. If the resulting set then violates the associated cut, it is added to the model. This is repeated until $|S| > limit$.
The detailed pseudocode for the algorithm is given in Section \ref{sec:clique-sep} of the appendix (Algorithm \ref{alg:general-clique}).

\FloatBarrier

\section{Experimental evaluation}\label{sec:experiment}
In our experimental evaluation, we are interested in answering the following questions:
\begin{itemize}
    \item Q1: How does the performance of our algorithms \ERS and \ERSTree compare to the state-of-the-art algorithms \flow and \tree~\cite{duque11,duque18} and to each other?
    \item Q2: How does the performance of each evaluated model depend on the predefined number of clusters $k$? 
\end{itemize}
The algorithm was implemented using \texttt{C++20} and compiled using \texttt{gcc 13} with \texttt{-O3} optimization flag. We used \texttt{Gurobi 12.0} as the ILP solver and the \texttt{Boost Graph Library 1.89 (BGL)} for various graph algorithms. Further implementation details can be found in Section \ref{section:implement-details} of the appendix.
\subsection{Benchmark Instances}
\textbf{Real-world instances.} Our motivation for studying the p-regions problem was its application in map aggregation, consequently we are interested in real-world geostatistical data for our computational experiments.  
In that regard, we follow the experimental section of the work by Oehrlein et al.~\cite{oehrlein17} and use data provided by the \emph{European Statistical Office (Eurostat)}~\cite{eurostat}. 
The data is acquired with respect to the NUTS (Nomenclature des unit\'es territoriales statistiques) subdivision of Europe~\cite{nuts-overview}. This subdivision contains three levels of hierarchy, with NUTS 1 being the coarsest and NUTS 3 the most fine-grained level of subdivision. 

The geostatistical data we use in particular is the unemployment rate in Europe from 2010~\cite{unemployment-data}, following Haunert et al.~\cite{oehrlein17}. Table \ref{tab:instances} gives an overview of the used instances. For all countries, we used the \texttt{NUTS-3} subdivison as the input subdivision, with the exception of \texttt{Germany}, where we used the \texttt{NUTS-2} subdivison as the size of the adjacency graph for \texttt{NUTS-3} was infeasible.
For vertices $u,v \in V$, the dissimilarity was calculated as $d(u,v) = |\gamma(u)-\gamma(v)|$, where $\gamma(u)$ is the unemployment rate of area $u$.
\begin{figure}[]
    \centering
    \begin{tabular}{ccccc}
         Country  & \# Vertices & \# Edges   \\
         \hline
         \texttt{Finland} & 19 & 41 \\
         \texttt{Bulgaria} &  28 & 60 \\
         \texttt{United Kingdom} & 36 & 74 \\
         \texttt{Greece} & 38 & 71 \\
         \texttt{Germany$^*$}  & 39 & 90 \\
         \texttt{Romania}  & 42 & 100 \\
         \texttt{Spain} &  47 & 111 \\
    \end{tabular}
    \caption{Overview of the instances used. The attribute used for each instance is the unemployment rate in each area. $^*$For \texttt{Germany}, we used \texttt{NUTS-2} as the input subdivision.}

    \label{tab:instances}
\end{figure}

\medskip\noindent
\textbf{Random instances.} Following the works of Duque et al.\ on the $p$-regions problem~\cite{duque11,duque18}, we include another set of randomly generated instances simulating spatial auto-regressive processes on grid graphs. For a given integer $\Size$ and $\rho \in [0,1]$, we generate a grid graph of size $\Size \times \Size$ with autocorrelation $\rho$ and the rook contiguity criterion~\cite{Lesage}. The details can be found in Section \ref{sec:random-generation} of the appendix.

We generated one instance for every combination of $\Size \in \{4,5,6,7,8,9\}$ and $\rho \in \{0,0.3,0.6,0.9\}$. Like for the real-world instances, the dissimilarity is calculated as $d(u,v) = |\gamma(u)-\gamma(v)|$, $\gamma(u)$ being the attribute of $u$.

\subsection{Test Setup}
We ran the experiments on an HPC cluster, with each node containing 2 × Intel Xeon "Sapphire Rapids" 48-core/96-thread 2.10GHz as CPUs and 1024GB of DDR5 4800MHz memory. The system is running the Linux distribution \texttt{AlmaLinux OS, Version 9.2 (Turquoise Kodkod)}. We set a time limit of 5 hours (18,000 s) for each instance.

\subsection{Results}
Figure \ref{fig:runtime} shows the number of solved instances plotted against the runtime for \ERS, \ERSTree, \flow and \tree.
The plots are grouped by different values for the number of predetermined regions $k$. Subfigures \ref{fig:runtime-rl1} - \ref{fig:runtime-rl4} show the results on the real-world instances, subfigures \ref{fig:runtime-rand1} - \ref{fig:runtime-rand4} on the random grid instances.

On the real-world instances (Subfigures \ref{fig:runtime-rl1} - \ref{fig:runtime-rl4}), we can see that our models \ERS and \ERSTree outperform both \tree and \flow for all values of $k$, especially for medium values ($k \in \{5,6,7,8\}$). In accordance with the work by Duque et al.~\cite{duque11}, we observe that \flow performs better for smaller values of $k$, while \tree performs better for larger values. For $k \in \{3,4\}$, \flow is somewhat competitive with \ERS and \ERSTree (although still slower), whereas for all other values of $k$, it falls far behind. Comparing \ERS and \ERSTree, we observe that in the majority of cases \ERSTree outperforms \ERS, which is explained by the superior relaxation strength of \ERSTree shown in Section \ref{sec:theory}. The performance improvement increases with higher values of $k$: While for $k=3$, both models show almost equal performance, \ERSTree clearly dominates for $k = \{7,8,9,10\}$. This coincides with the fact that \tree performs much better with increasing $k$. We therefore hypothesize that the performance impact of the subtour elimination constraints increases with the number of regions. 

\begin{figure}
\centering

\begin{minipage}{0.45\linewidth}
\centering
\includegraphics[width=\linewidth]{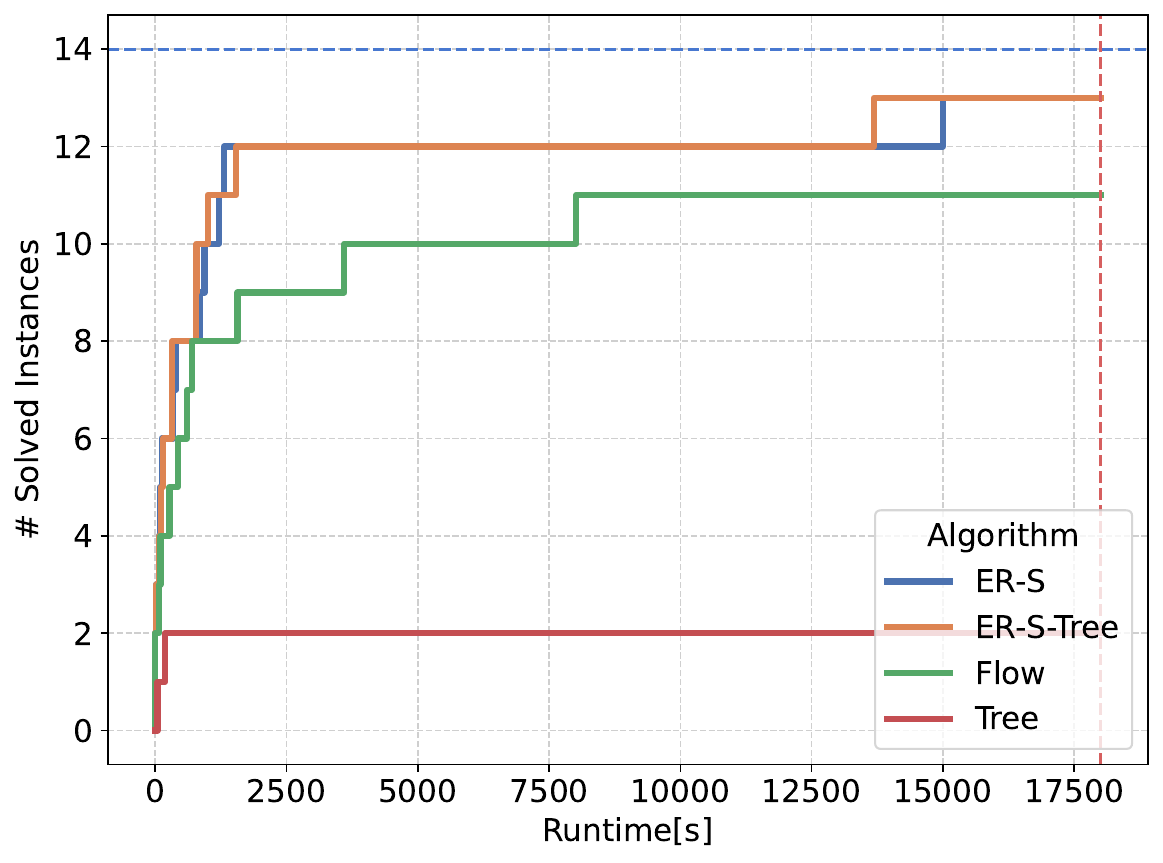}
\subcaption{Results for $k \in \{3,4\}$, real-world instances.\label{fig:runtime-rl1}}
\end{minipage}
\hfill
\begin{minipage}{0.45\linewidth}
\centering
\includegraphics[width=\linewidth]{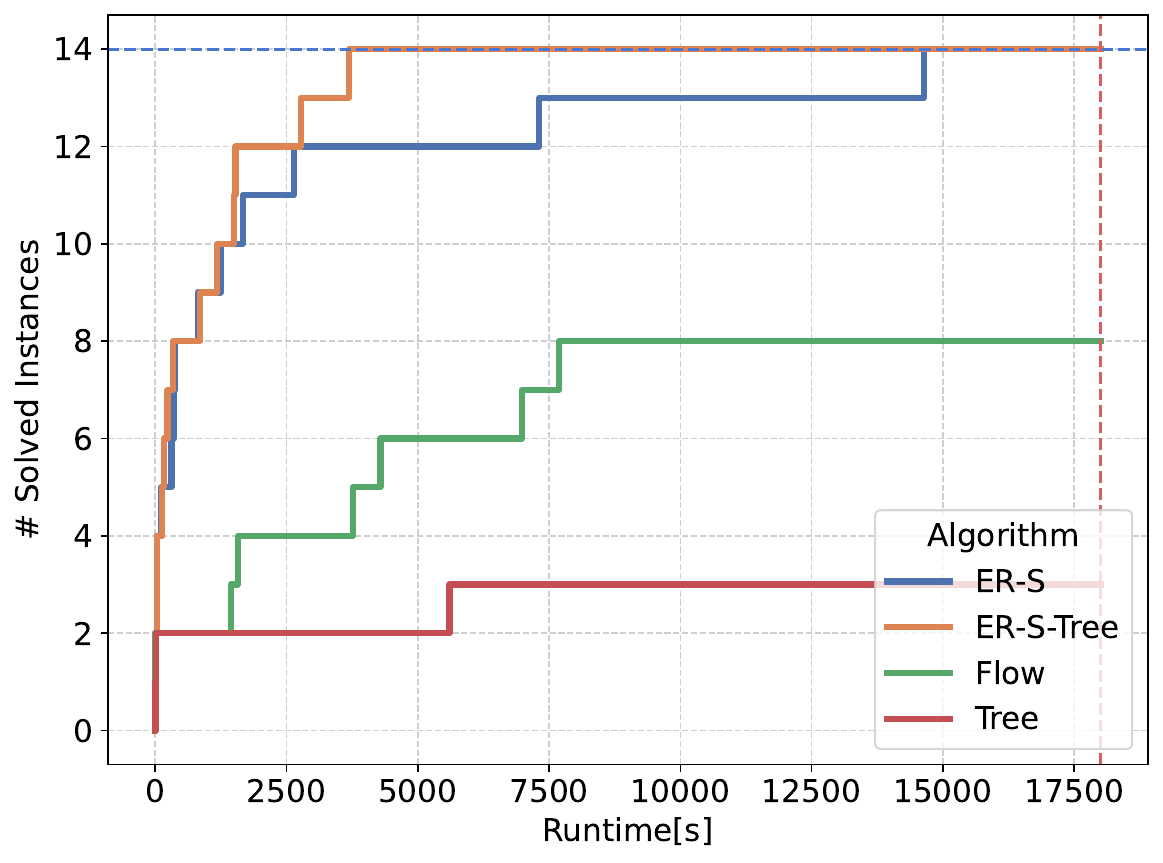}
\subcaption{Results for $k \in \{5,6\}$, real-world instances.\label{fig:runtime-rl2}}
\end{minipage}
\centering
\begin{minipage}{0.45\linewidth}
\centering
\includegraphics[width=\linewidth]{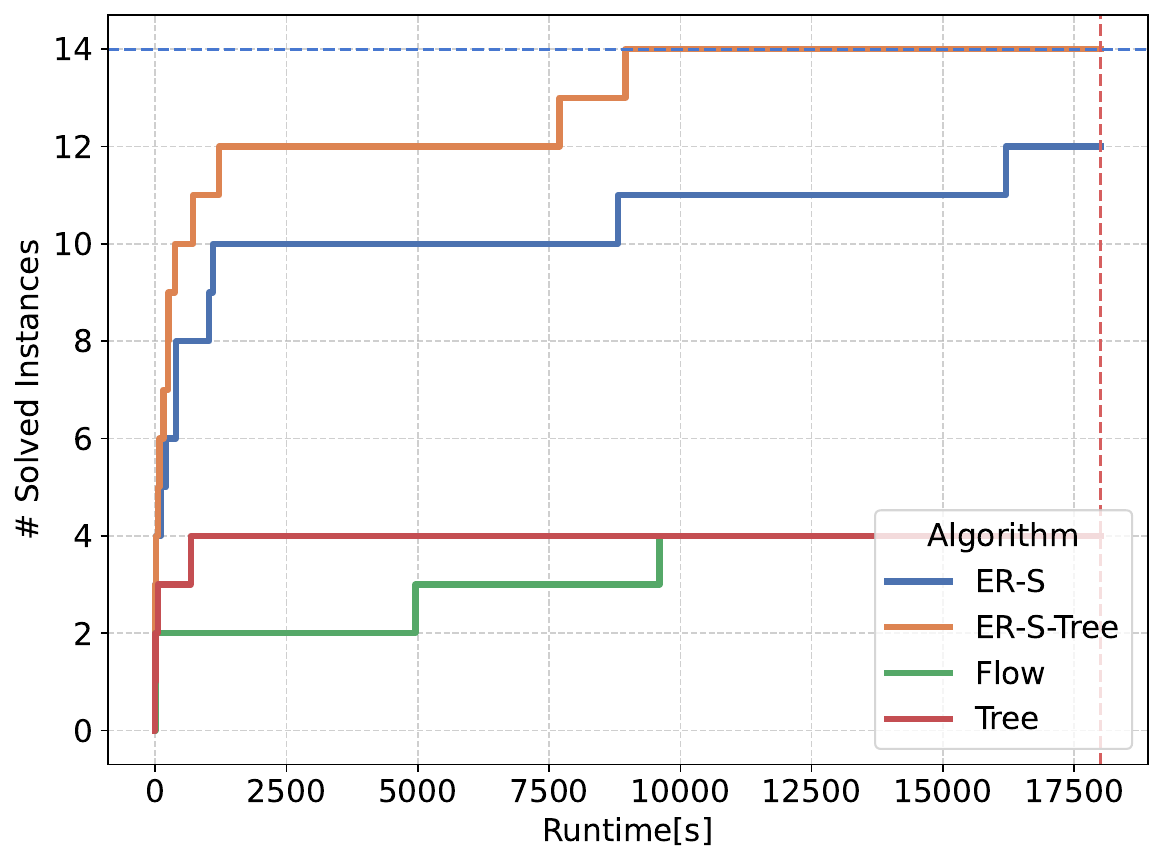}
\subcaption{Results for $k \in \{7,8\}$, real-world instances.\label{fig:runtime-rl3}}
\end{minipage}
\hfill
\begin{minipage}{0.45\linewidth}
\centering
\includegraphics[width=\linewidth]{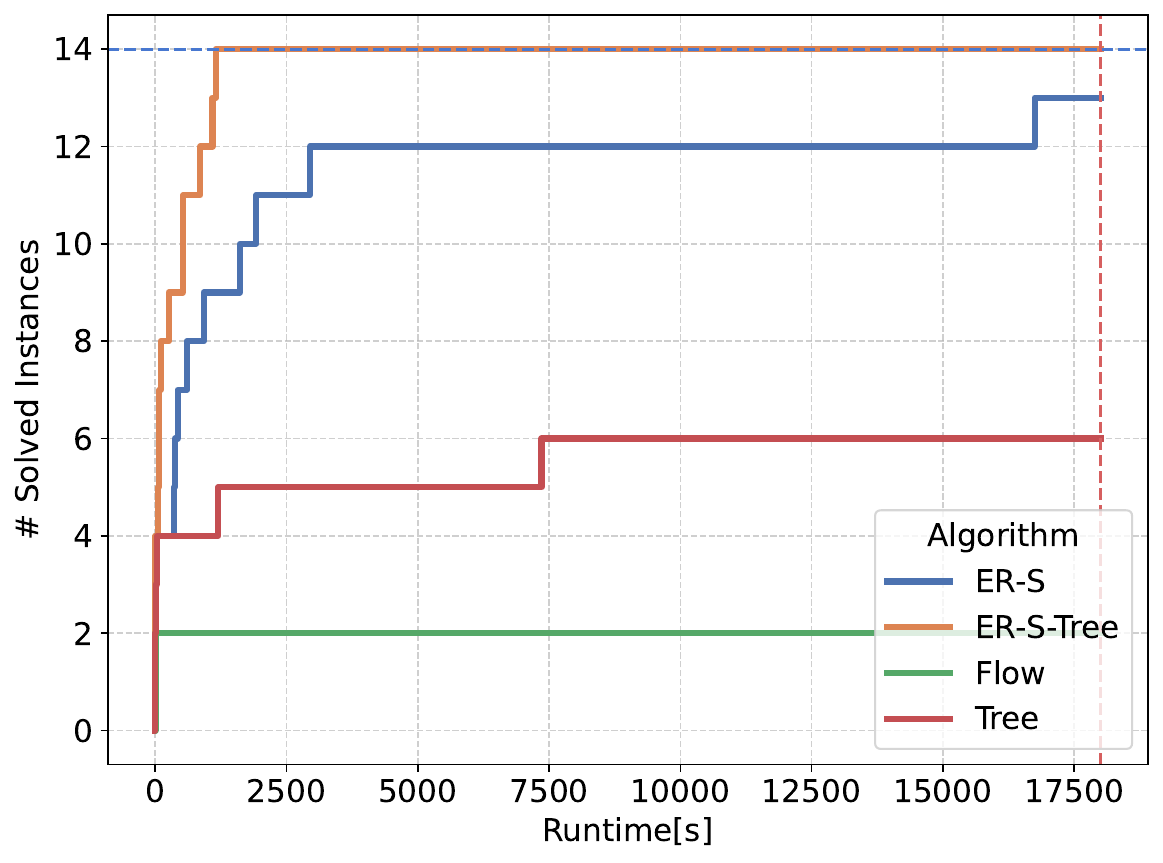}
\subcaption{Results for $k \in \{9,10\}$, real-world instances.\label{fig:runtime-rl4}}
\end{minipage}

\centering

\begin{minipage}{0.45\linewidth}
\centering
\includegraphics[width=\linewidth]{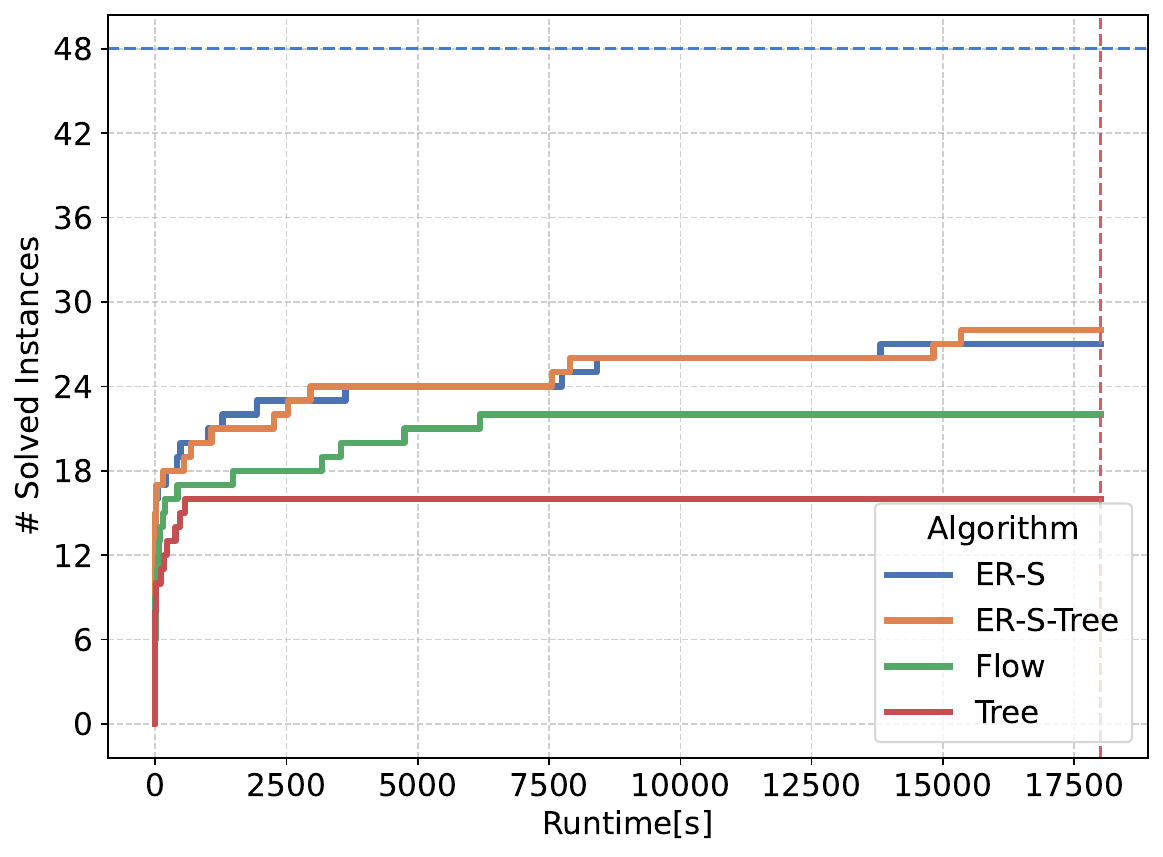}
\\
\subcaption{Results for $k \in \{3,4\}$, grid instances.\label{fig:runtime-rand1}}
\end{minipage}
\hfill
\begin{minipage}{0.45\linewidth}
\centering
\includegraphics[width=\linewidth]{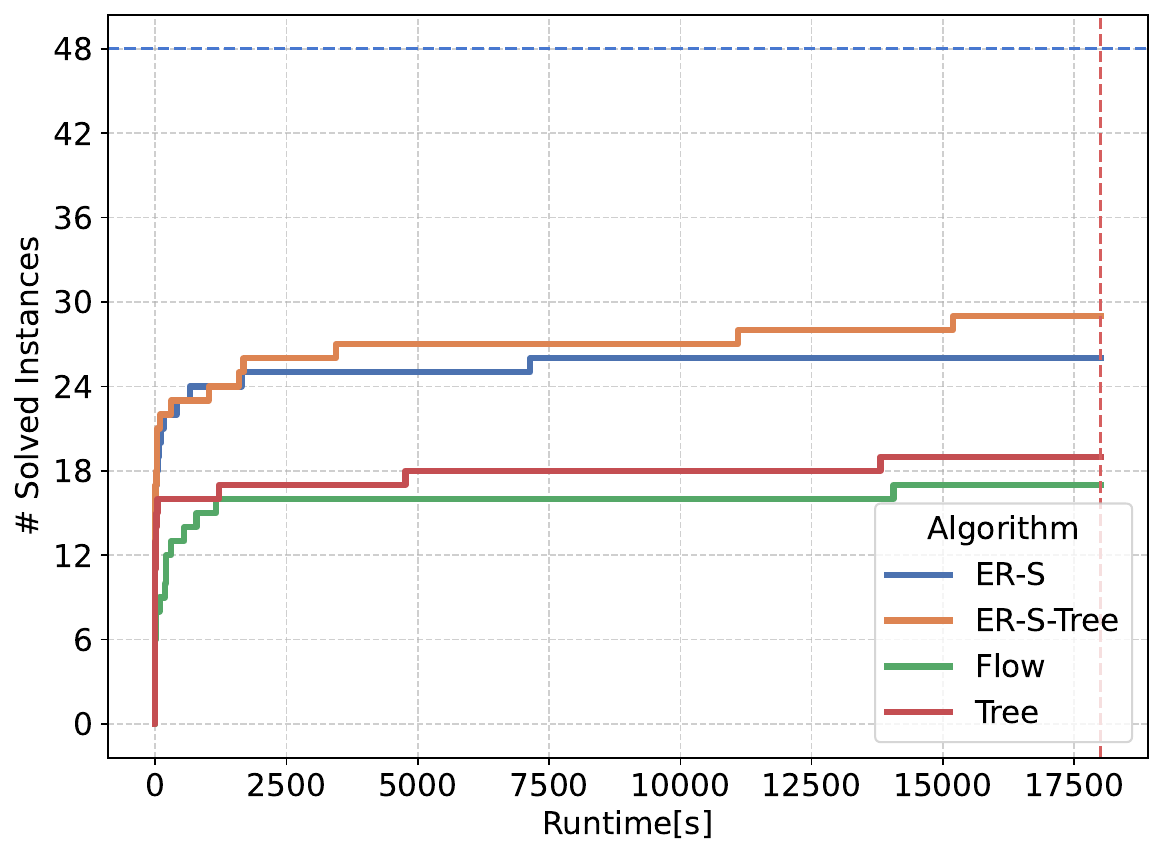}
\\
\subcaption{Results for $k \in \{5,6\}$, grid instances.\label{fig:runtime-rand2}}
\end{minipage}
\centering
\begin{minipage}{0.45\linewidth}
\centering
\includegraphics[width=\linewidth]{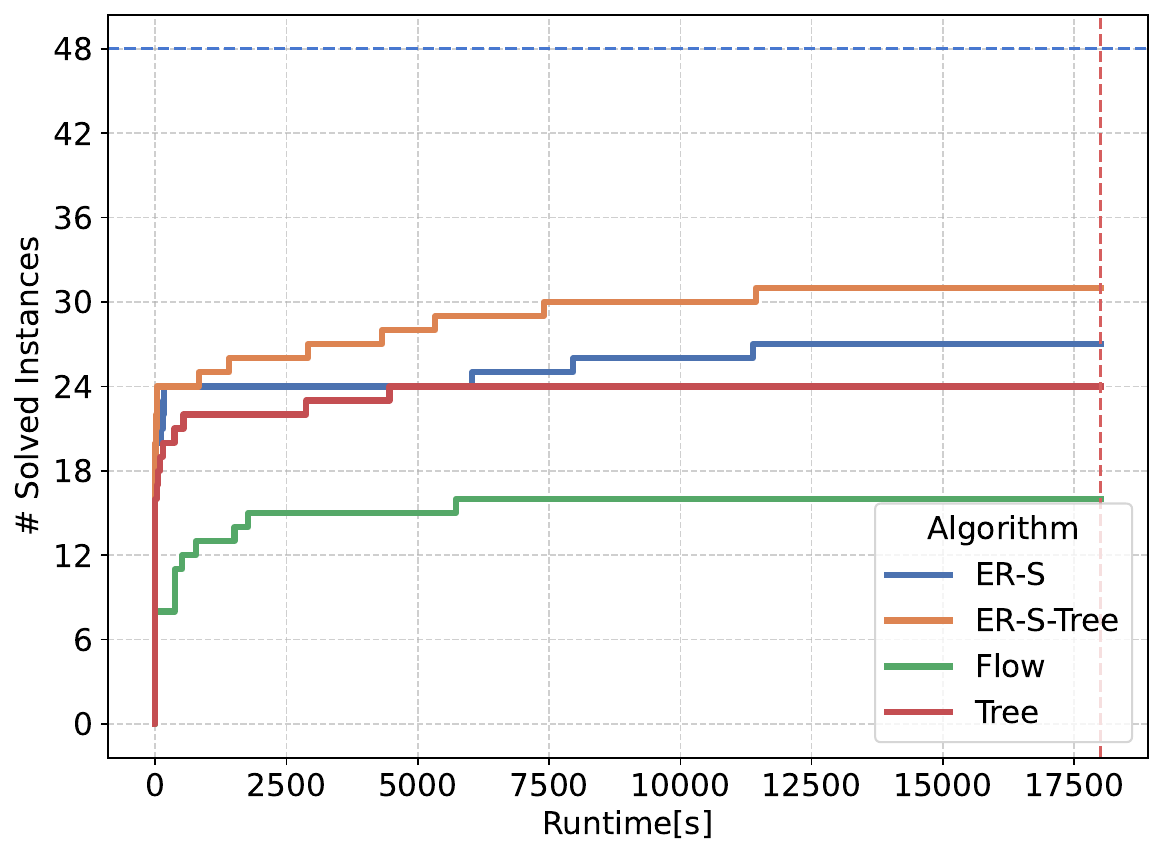}
\\
\subcaption{Results for $k \in \{7,8\}$, grid instances.\label{fig:runtime-rand3}}
\end{minipage}
\hfill
\begin{minipage}{0.45\linewidth}
\centering
\includegraphics[width=\linewidth]{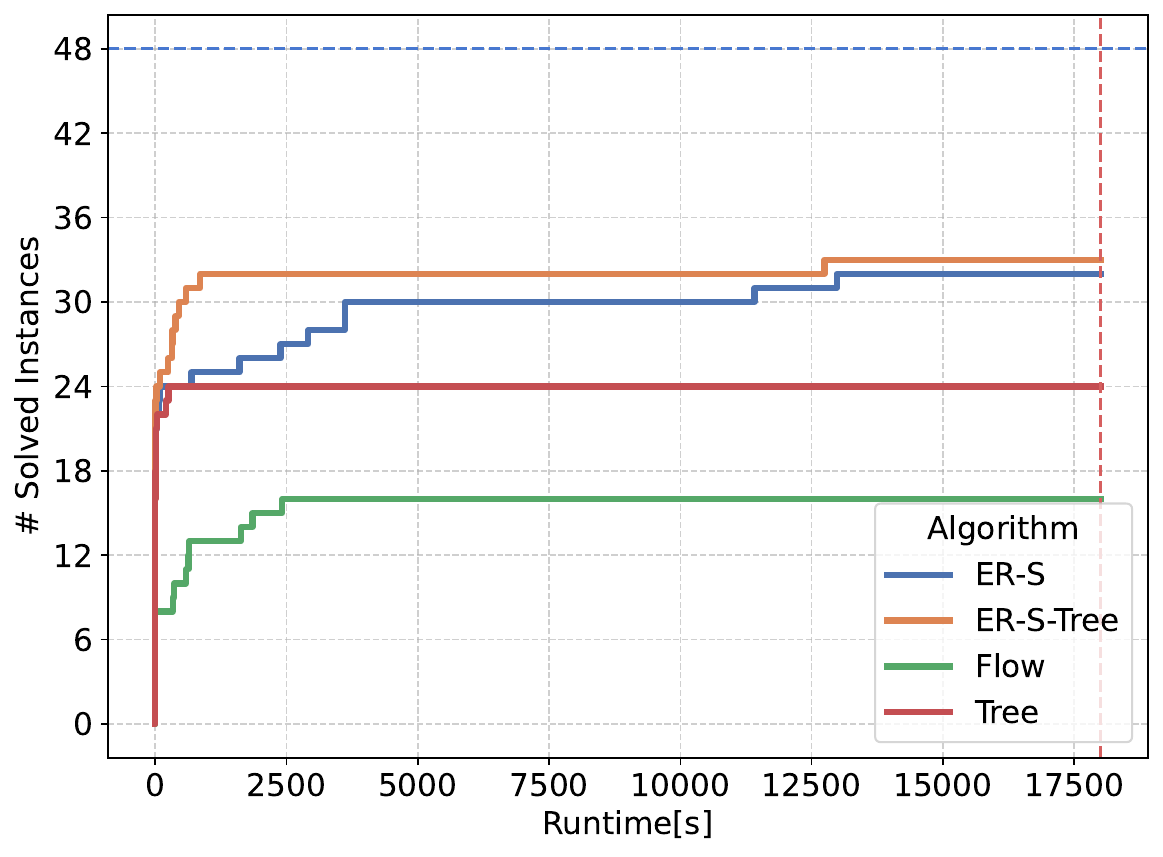}
\\
\subcaption{Results for $k \in \{9,10\}$, grid instances.\label{fig:runtime-rand4}}
\end{minipage}
\caption{Performance plot showing the number of solved instances within a given time for \ERS, \ERSTree, \flow and \tree, grouped by the number of predefined numbers $k$.}
\label{fig:runtime}
\end{figure}

For the generated grid instances (Subfigures \ref{fig:runtime-rand1} - \ref{fig:runtime-rand4}), we observe similar performance trends as for the real-world instances; \ERS and \ERSTree outperform the models \flow and \tree for all choices of $k$, while \flow outperforms \tree for low values of $k$ and \tree outperforms \flow for high values of $k$. \ERSTree outperforms \ERS for higher values of $k$, while for lower values the performance is similar. One notable difference is that for the grid instances, \tree has better performance relative to \flow when compared to the real-world instances. It seems that \tree benefits from a grid graph structure.

We mention that for instances that were not solved by any model. i.e. grid instances of size $8 \times 8$ and $9\times 9$, \ERS and \ERSTree consistently provide much stronger bounds than \flow and \tree, in many cases achieving optimality gaps of $< 10 \%$. The detailed results can be found in the full version of the paper footnote{See \url{https://arxiv.org/abs/2607.04886}}. For an ablation study of the performance impact of different implementation choices, see Section \ref{sec:ablation-study} of the appendix.

\section{Conclusion and Outlook} \label{sec:conclusion}
In this paper, we proposed two new ILP models for the $p$-regions problem: \ERS which models connectivity using vertex separator constraints, and \ERSTree, which combines vertex-separator constraints and subtour elimination constraints. The model \ERSTree also includes a new type of subtour elimination constraint, which is specific to the $p$-regions problem. We conducted a theoretical analysis of the polyhedral strength of each model, proving that \ERSTree is strictly stronger than the other models. Our computational experiments show that the proposed models \ERS and \ERSTree strongly outperform the models from the literature across all chosen values for $k$, both on real-world instances and generated grid instances, supporting our theoretical results. 

In line with other work~\cite{oehrlein17,LiCG14}, we would like to incorporate regional compactness into the problem objective in the future.
Notably, the popular compactness criterion of minimizing cut edges~\cite{ValidiB22} can be incorporated without any change to the algorithms; it simply suffices to modify the costs for every vertex pair that is connected by an edge. First experiments showed that including compactness into the objective often decreases the running time of the algorithms, as connected regions emerge more naturally. Figure \ref{fig:regions-compact} shows two subdivisions of Spain for different objectives.
\begin{figure}
\begin{minipage}{0.44\linewidth}
\centering
\includegraphics[width=\linewidth]{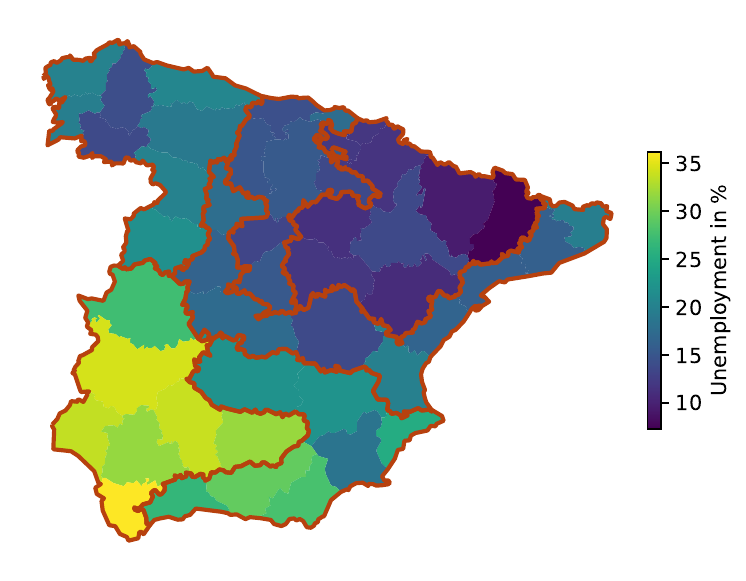}
\subcaption{Only considering pairwise feature similarity.} 
\vfill
\end{minipage}
\hfill
\begin{minipage}{0.44\linewidth}
\centering
\includegraphics[width=\linewidth]{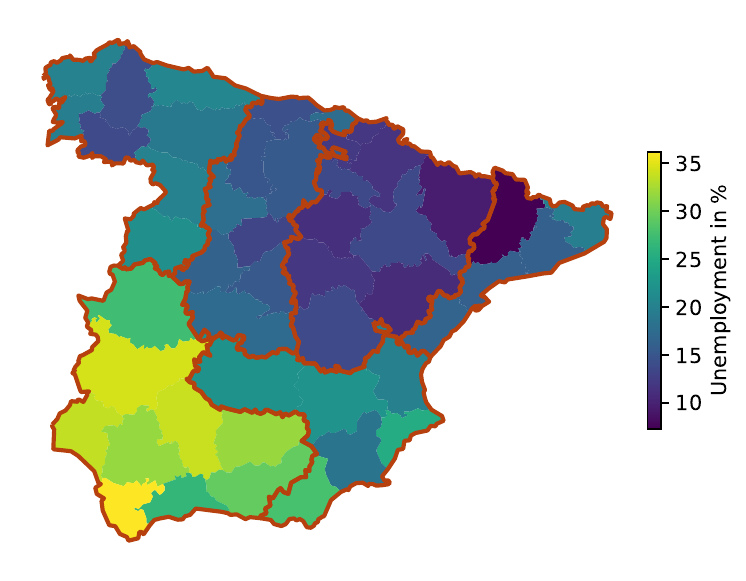}
\subcaption{Linear combination between feature similarity and cut edges.}
\end{minipage}
\caption{Examples of resulting regionalization on the NUTS-3 subdivison of Spain for $k = 6$ for different costs}
\label{fig:regions-compact}
\end{figure}

For further research, it would be interesting to apply techniques from this paper to the max-$p$-regions problem ~\cite{duque12}. In this variant of the problem, the number of output regions is not fixed, and the granularity of the partition is instead constricted by bounds on the region sizes. 
As a general direction, our theoretical results show that vertex separator constraints and subtour elimination constraints cut off different fractional solutions from the solution polytope. Combining separator constraints and subtour elimination constraints therefore seems to be a promising direction for a wider range of connectivity constrained problems.
\bibliography{bibliography}
\FloatBarrier
\appendix
\section{Proofs of polyhedral results}
\subsection{Proof of Theorem \ref{theorem:flow-relax}}\label{proof:flow-relax}
\begin{proof}
    For $k=n$, simply set $z_{a,a} = 1$ and $r_{a,a}$ and all flow variables to zero. The right-hand side of \eqref{eq:flow2} is always zero and therefore satisfied.

    For $3 \leq k \leq n-1$, let $(z,r,x,f) \in P(\flow)$ for $G,k$. For any $i \neq 1$ and $v \neq 1$, set $z_{v,i} = \frac{1}{k-1}$ and set $r_{a,i} = \frac{1}{n-1}$ and all flow variables to zero. This assignment trivially satisfies constraints \eqref{eq:flow3}-\eqref{eq:flow7}, \eqref{eq:flow9} and \eqref{eq:flow2} for any value of $x_{u,v}$ (the right-hand side is always $<0$). 
    
    For constraints \eqref{eq:flow8}, the following holds:
    \begin{equation*}
        M\cdot r_{v,i} = \frac{n-k+1}{n-1} = 1-\frac{k-2}{n-1} \geq 1-\frac{k-2}{k} = \frac{2}{k} \geq \frac{1}{k-1}
    \end{equation*}
    where the last inequality follows from $k \geq 3$. Therefore, $z_{v,i} - M \cdot r_{v,i} \leq 0$ and the constraints are always satisfied.
    
    It immediately follows that $P_x(\ERS) \subseteq P_x(\flow)$ and $P_x(\tree) \subseteq P_x(\flow)$. We now show that these inclusions are strict for $k < n$: Let $G$ be a any graph with $V = \{v_1,\dots, v_n\}$ and $3 \leq k < n$ and let $z_{a,l} = \frac{1}{k-1}$, $r_{a,l} = \frac{1}{n-1}$ for $a \in V \setminus\{1\},\; l \geq 2$ and $x_{u,v} = 0$ for $u,v \in V, u < v$ be a solution for \flow.
    This solution is neither contained in \ERS nor \tree: For \ERS, constraint \eqref{eq:ERS5} implies $r_v < 1$ for at least one vertex $v$ (as $k < n$). By constraints \eqref{eq:ERS4}, this implies $\sum_{u < v}x_{u,v} >0$, contradicting $x_{u,v} = 0$ for all $u < v$. Furthermore, this solution is also not contained in \tree, as constraints \eqref{eq:tree4} enforce that $\sum_{e \in \diedge{E}}y_e \leq \sum_{u < v}x_{u,v} = 0$, which contradicts Constraint \eqref{eq:tree1}.
\end{proof}
\subsection{Proof of Theorem \ref{theorem:treestrong}}\label{proof:treestrong}
\begin{proof}
    Obviously, it holds that $P_x(\treeplus) \subseteq P_x(\tree)$ as constraints \eqref{eq:tree8} imply \eqref{eq:tree4} and all other constraints of \tree are also part of $\treeplus$. To show that the strengthened cycle inequalities \eqref{eq:tree-strong} can cut off additional fractional solutions, consider the following instance: Let $G$ be a triangle graph with $V=\{v_1,v_2,v_3\}$ and $k=1$. A valid solution for the fractional relaxation of \tree is $x_{1,2} = x_{2,3}=x_{1,3} = y_{1,2} = y_{2,3}=y_{3,1} = \frac{2}{3}$. However, this solution violates constraints \eqref{eq:tree-strong}, as 
    \begin{equation*}
         2\cdot\sum_{(u,v) \in \diedge{E}}y_{u,v} - \sum_{\{u,v\} \in E}x_{u,v} = 4 -2 \not \leq |V|-2
    \end{equation*}
\end{proof}
\subsection{Proof of Observation \ref{obs:tree-strong}}\label{proof:tree-strong}
\begin{proof}
Let $S\subseteq V$ and $T \subseteq E(S)$ with $|T| = |S|$ and let $(x^*,y^*)$ be a variable assignment satisfying Constraint \eqref{eq:tree-strong} for $S$. Then it holds that:
\begin{align}
    2\cdot\sum_{(u,v) \in \diedge{T}}y^*_{u,v} &= 2\cdot\sum_{(u,v) \in \diedge{T}}y^*_{u,v} - |S|+|S|\\
    &\leq 2\cdot\sum_{(u,v) \in \diedge{T}}y^*_{u,v} - \sum_{\{u,v\} \in T}x^*_{u,v} + |S| \\
    &\leq 2\cdot(|S|-1)
\end{align}
\end{proof}
\subsection{Proof of Theorem \ref{theorem:erstree-incomp}}\label{proof:erstree-incomp}
The example consists of the graph with $V = \{v_1, v_2, v_3\}$ $E = \{\{v_1,v_2\},\{v_2,v_3\}\}$. Consider the following solution $(x,y) \in P(\tree)$: Let $k = 2$, $x_{v_1,v_2} = x_{v_2,v_3} = \frac{1}{2}$, $x_{v_1,v_3} = 1$ and $y_{v_1,v_2} = y_{v_2,v_3} = \frac{1}{2}$. The feasibility of this solution is trivial, however there exists no feasible solution $(x,r) \in P(\ERS)$. The reason for this is the violation of the separator constraint $x_{v_1,v_3} \not \leq x_{v_1,v_2}$.

\begin{figure}
    \centering
    \includegraphics[]{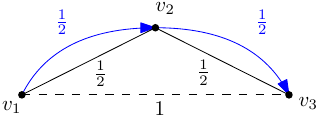}
    \caption{Example for $P(\tree) \not \subset P(\ERS)$: $k=2$, the values of $y$ are in blue, the values of $x$ in black. Solid lines depict existing edges}
    \label{fig:seperator-implied-counter2}
\end{figure}

To show that $P_x(\ERS) \not \subset P_x(\tree)$, we can use the example depicted in Figure \ref{fig:cycle-implied-counter2} for $k=4$, which is in $P(\ERS)$ but not $P(\tree)$. First we show that the solution is feasible for \ERS: Constraints \eqref{eq:ERS2} and \eqref{eq:separator} are satisfied as the vertex sets $\{v_1,v_2,v_3,v_4\}$ and $\{v_5,v_6,v_7,v_8\}$ all have the same pairing values within each group and all pairings between the two groups are zero. Furthermore, the values of the representative variables satisfy \eqref{eq:ERS3}-\eqref{eq:ERS4} and sum up to exactly 4. The pairing variables of adjacent edges sum up to $4$, so \eqref{eq:edge-sum} is also satisfied. Finally, we consider the general clique constraints \eqref{eq:gen-clique}: For $|Q|=5$, the vertex set with the lowest clique sum consists of $Q_5 = \{v_1,v_2,v_3,v_4,v_5\}$, which sums up to $\frac{6}{5} \geq 1$. For $|Q| = 6$, the lowest sum is obtained with $Q_6$ being the union of $Q_5$ and one vertex from $v_6,v_7,v_8$, which sums up to $2$. For $Q_8$, the total sum of all pairing variables is $6 \geq 4$ and for $Q_7$, it must be at least $6 - \frac{12}{5} = \frac{18}{5} \geq 3$. To see that there cannot be a valid solution for \tree with such an assignment of pairing variables, observe that $\sum_{\{u,v\} \in E}x_{u,v} = 4 = n-k$, so for $\sum_{(u,v) \in \diedge{E}}y_{u,v} = n-k$ to hold, it must hold that $y_{u,v}+y_{v,u} = x_{u,v}$ for all $\{u,v\} \in E$. However, for the cycle $C = \{v_5,v_6,v_7,v_8\}$, this would imply that $\sum_{(u,v) \in \diedge{E}(C)}y_{u,v} = \frac{16}{5} > 3$, which would violate the cycle inequality.

\begin{figure}
    \centering
    \includegraphics[]{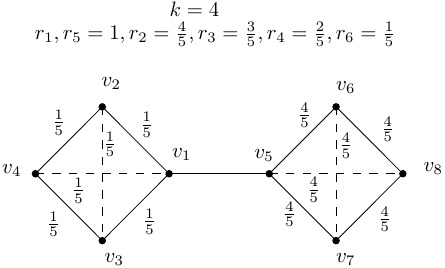}
    \caption{Example for $P(\ERS) \not \subset P(\tree)$, $k=4$: the figure depicts the non-zero values of $x$, existing edges are drawn as solid lines}
    \label{fig:cycle-implied-counter2}
\end{figure}

$P_x(\tree) \not \subset P_x(\ERS)$ and $P_x(\ERS)\not \subset P_x(\treeplus)$ follow from $P_x(\treeplus) \not \subset P_x(\ERS)$ and $P_x(\ERS)\not \subset P_x(\tree)$ and Theorem \ref{theorem:treestrong}.

\section{Separation of General Clique Constraints}\label{sec:clique-sep}
In this section, we give the pseudocode of the heuristic we use to separate constraints \eqref{eq:gen-clique}. We denote $\bar{x}(u,v)$ to be the fractional value of $x_{u,v}$ in the current solution and define $\bar{x}(u,S) = \sum_{v \in S}x_{u,v}$ for $S \subseteq V \setminus \{u\}$.
\begin{algorithm}
\caption{Seperating General Clique Constraints}\label{alg:general-clique}
\begin{algorithmic}
\Require Fractional assignments $\bar{x}:{V\choose 2} \rightarrow \mathbb{R}_{\geq 0}$, Clique size limit $limit$
\State Find $u,v \in V$ minimizing $\bar{x}(u,v)$
\State $S \leftarrow \{u,v\}$
\While{$|S| \leq limit$}
\State Let $w \in V \setminus S$ be the vertex minimizing $\bar{x}(w,S)$ 
\State $S \leftarrow S \cup \{w\}$ \Comment{Greedily extending $S$ with the best vertex}
\State Find vertex pairs $a\in V \setminus S$, $b \in S$ minimizing $\Delta(a,b) = \bar{x}(a,S)- \bar{x}(b,S)-\bar{x}(a,b)$
\If{$\Delta(a,b) < 0$}
\State $S \leftarrow (S \setminus \{b\}) \cup \{a\}$ \Comment{Swap $b$ for $a$}
\EndIf
\State Let $|S| = qk+p$
\If{$|S| \geq k+1$ \textbf{and} $\bar{x}(E(S)) < {q+1 \choose 2}p+{q \choose 2}(k-p)$}
\State Add inequality $x(E(S)) \geq {q+1 \choose 2}p+{q \choose 2}(k-p)$
\EndIf
\EndWhile
\end{algorithmic}
\end{algorithm}
\FloatBarrier
\section{Implementation Details}\label{section:implement-details}
In this section, we give further details on our implementation.
\paragraph*{Separation routines}
We used the implementation of the Boykov-Kolmogorov flow algorithm~\cite{BoykovK04} provided in \texttt{BGL} as a subroutine for computing the minimum cut in the separation of constraints \eqref{eq:separator} and \eqref{eq:tree4}.

For fractional solutions, we only add cuts if they are violated by at least a tolerance $\epsilon$, i.e. a cut $a^tx \leq b$ is added if $a^tx > (1+\epsilon) b$ or in the case of $b = 0$ if $a^tx > \epsilon$. We set $\epsilon = 5\%$

For the general clique inequalities, we set $limit = 3\cdot k+3$. An exception is the general clique inequality corresponding to $S=V$, which we add beforehand. We observed that this helps finding a strong initial root relaxation.

To add cutting planes during the solving process, we use the callback function of Gurobi. We add violated inequalities using \texttt{addLazy} in integer separation routines, and \texttt{addCut} in fractional separation routines. 
The difference is that constraints added using \texttt{addLazy} are enforced in each branch-and-bound node, while constraints added using \texttt{addCut} are stored in the global cut pool and are applied selectively by the solver.

Despite the internal cut filters of the solver, we still observed slightly better performance if we prefiltered our cuts before handing them to the solver. 
For this reason, during the fractional separation routines, we only add the $30\%$ of identified violated cuts that have the highest efficacy~\cite{achterberg2009scip} (the distance of the hyperplane from the separated invalid point). In the case of integer solutions, all violated separator constraints are added.
\paragraph*{Reimplementation of algorithms}
The authors in~\cite{duque11} did not provide the source code for their algorithms, thus we reimplemented the algorithms following the descriptions in~\cite{duque11} and the follow-up paper~\cite{duque18} as closely as possible. For \texttt{Tree}, we used \texttt{addLazy} to separate violated cycle-breaking inequalities in each integer callback. Following~\cite{duque18} we used the \texttt{BGL} implementation of Tarjan's algorithm to calculate strongly connected components. For each strongly connected component $\Gamma$ containing more than one vertex, we add the corresponding inequality of type \eqref{eq:tree4}.

We used default parameters for \texttt{Gurobi} (with the exception of setting \texttt{PreCrush=1}, as this is necessary to guarantee user cuts working correctly).

\section{Generation of random instances}\label{sec:random-generation}
For the generation of the grid instances, we generated the vertex attributes simulating a spatial auto-regressive model (see for example \cite{hoef18}): Given an integer $Size$ and an autocorrelation parameter $\rho$, the generated graph is a grid graph of size $Size \times Size$. The vertex attributes are given as $\Gamma_{G,\rho} = ((I-\rho \bar{A})^{-1})^t\epsilon$, where $I$ is the identity matrix, $\bar{A}$ is the row-normalized adjacency matrix of $G$, and $\epsilon$ is a vector of entries drawn from a normal distribution. The vertex dissimilarities between vertices $u$ and $v$ are now defined as $d(u,v) = |(\Gamma_{G,\rho})_u-(\Gamma_{G,\rho})_v|$, $(\Gamma_{G,\rho})_u$ representing the matrix entry corresponding to vertex $u$.

\section{Ablation study}\label{sec:ablation-study}
To determine the impact of various implementation details on the performance of \ERS and decide on the final settings for our experimental study, we performed some preliminary experiments on a limited subset of the benchmark set. The instances used in these experiments include \texttt{Finland}, \texttt{Bulgaria}, \texttt{UK} and \texttt{Spain} tested for $k \in \{3,5,7,9\}$.
\subsection{Size limit for separated general clique cuts}
We tested the performance of \ERS for various values of $limit = \alpha \cdot k + \beta$, where $limit$ is the size limit of separated cliques in Algorithm \ref{alg:general-clique} and $\alpha,\beta \in \mathbb{N}$.  
Figure \ref{fig:plot-clqcuts} shows the results for $\alpha,\beta \in \{1,2,3\}$. We observe that setting $\alpha = \beta = 1$ clearly shows the worse performance, therefore the solver benefits from separating \textbf{general} clique constraints (in contrast to just separating regular clique constraints). Beyond that, all configurations with $\alpha \in \{2,3\}$ exhibit very similar performance. For our final experiments, we chose $\alpha=\beta =3$.
\begin{figure}
    \centering
    \includegraphics[width=0.5\linewidth]{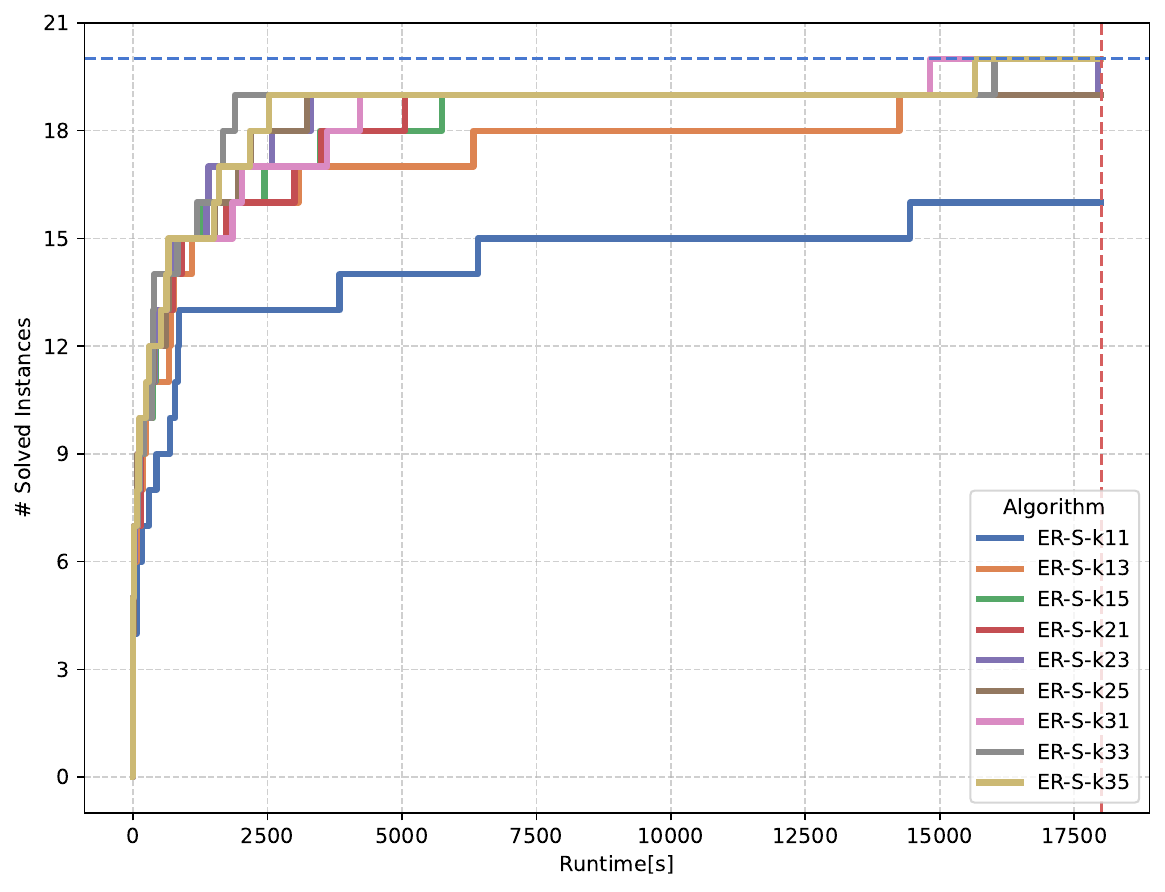}
    \caption{Size limit for general clique constraints}
    \label{fig:plot-clqcuts}
\end{figure}

\subsection{Impact of implemented features}
We evaluated the performance impact of our implementation decisions on the limited benchmark set, additionally including the instance \texttt{Romania}. For each of the tested feature, we ran the model with this specific feature deactiviated. The tested features where the following:
\begin{itemize}
    \item (Leaving out) Filter: Filtering out the cuts with the highest efficacy was deactivated, i.e. all identified cuts where added
    \item (Leaving out) Local Search (LocSrch): The local search step in the separation of the general clique inequalities was skipped, i.e. the separation algorithm exclusively performs the greedy extension step in each iteration
    \item (Leaving out) Nested Cuts (Nested): We do not perform nested cuts, only adding at most one cut to each vertex pair in each iteration
    \item (Leaving out) strengthening inequalities (SI): We omit adding \eqref{eq:edge-sum} and the general clique inequality for $S=V$ to the initial model
\end{itemize}
The results are shown in Figure~\ref{fig:plot-ablation}. \texttt{ER-S-Full} denotes \ERS, i.e. all features enabled. We can see that for the relatively easy instances, the performance is close to equal, while for the harder instances, we see some notable differences in performance. Disabling the local search step in the general clique separation seems to yield the largest performance degradation, followed by disabling strengthening inequalities and cut filtering. On the other hand, the performance benefit of using nested cuts is not completely clear, as there are multiple instances that are solved faster when nested cuts are disabled. Nevertheless, we still chose the current setting as the full model solves the most instances within a 5 hour time limit.
\begin{figure}
    \centering
    \includegraphics[width=0.45\linewidth]{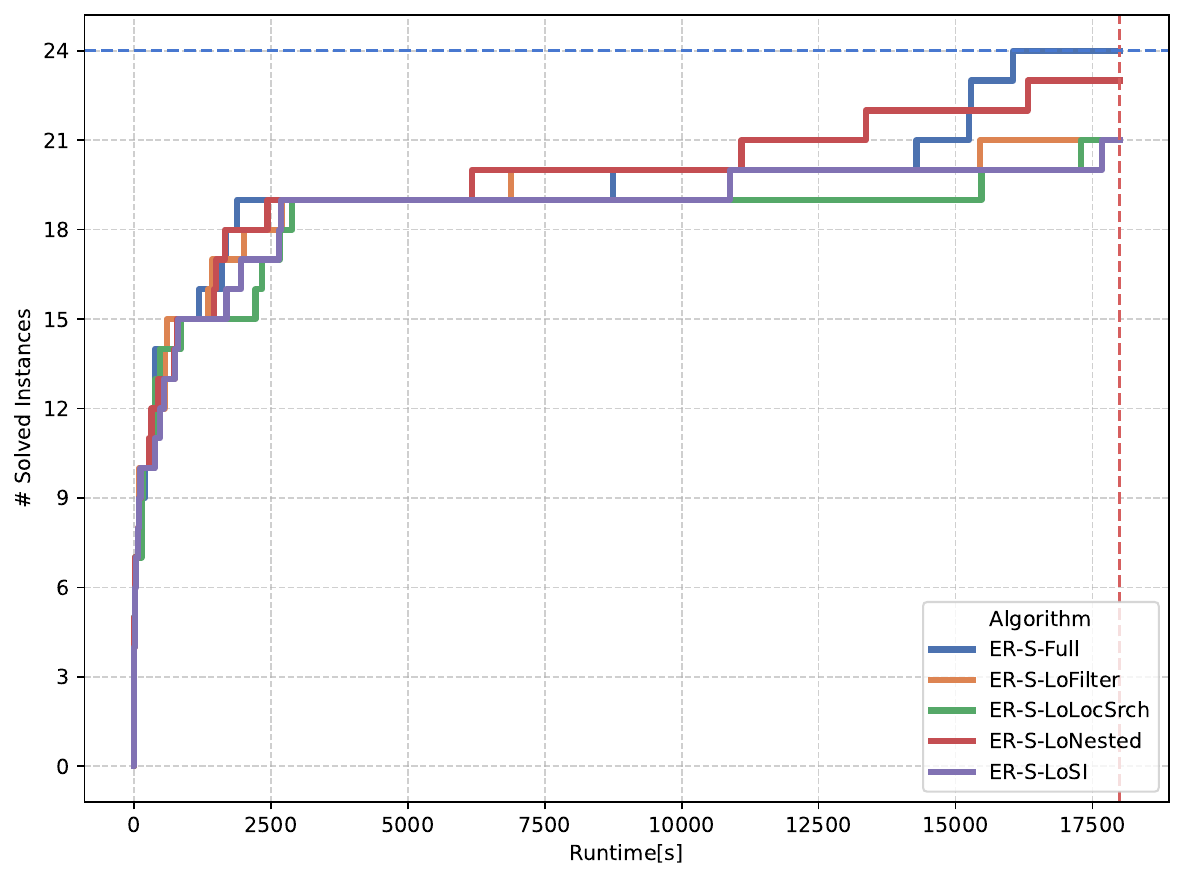}
    \caption{Leave-one-out experiment for different implemented features}
    \label{fig:plot-ablation}
\end{figure}

\subsection{Experimental comparison of \tree and \treeplus}
We conducted a computational comparison between \tree and \treeplus on the limited benchmark set. The results can be seen in Figure \ref{fig:treeplus}. We can see that \treeplus outperforms \tree, solving two more instances.
\begin{figure}
    \centering
    \includegraphics[width=0.45\linewidth]{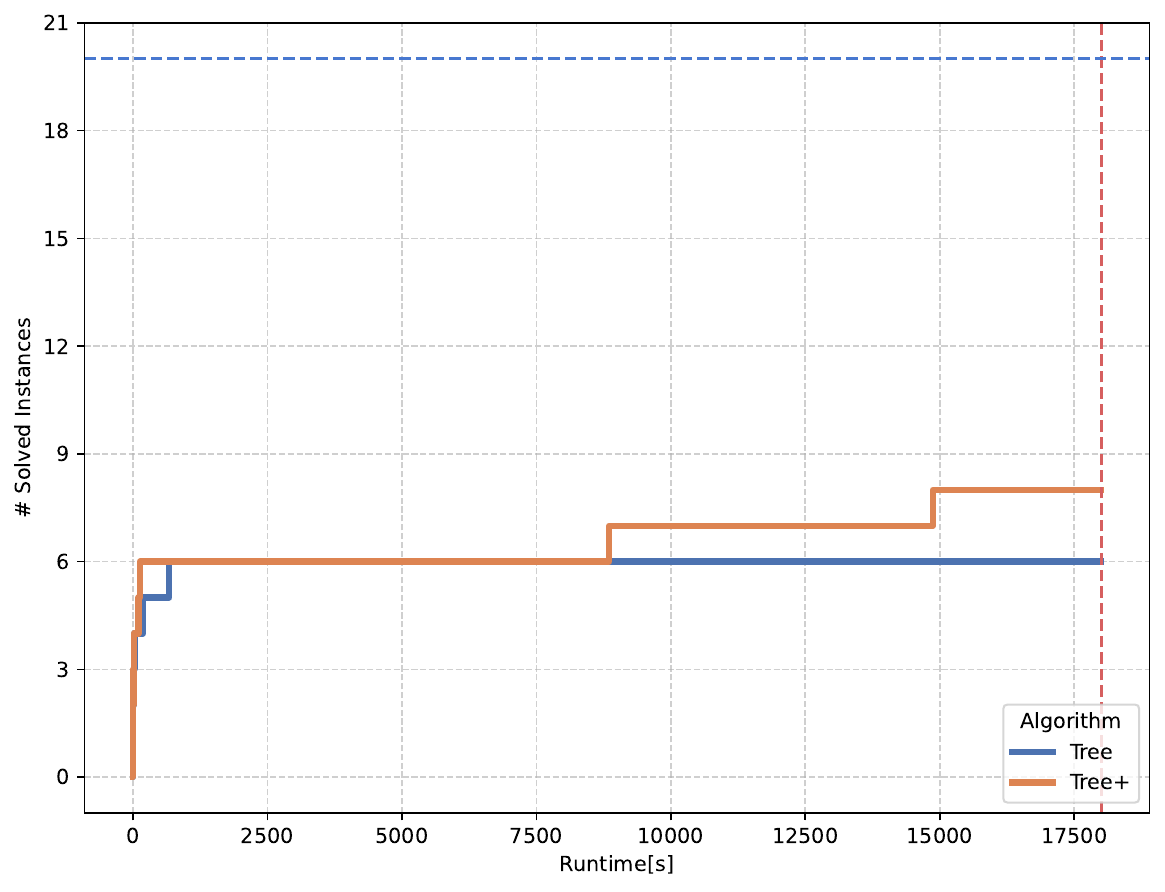}
    \caption{Performance comparison of \tree and \treeplus}
    \label{fig:treeplus}
\end{figure}
\section{Detailed runtime results}\label{sec:detailed-results}
\begin{table}
\caption{Detailed runtime results on real-world instances}
\label{tab:results}
\centering
\tiny
\begin{tabular}{|ll||r|r|r|r|}
\toprule
 & & \flow & \tree & \ERS & \ERSTree \\
Instance & $k$ &  &  &  &  \\
\midrule
\multirow[c]{8}{*}{\texttt{Finland}} & 3 & 0.99 & 188.16 & 1.17 & 0.93 \\
 & 4 & 3.04 & 44.18 & 0.70 & 0.74 \\
 & 5 & 5.37 & 17.16 & 0.56 & 0.86 \\
 & 6 & 3.97 & 6.01 & 0.57 & 0.32 \\
 & 7 & 4.19 & 3.78 & 0.70 & 0.59 \\
 & 8 & 6.80 & 2.23 & 0.22 & 0.24 \\
 & 9 & 5.24 & 1.89 & 0.20 & 0.11 \\
 & 10 & 3.51 & 0.59 & 0.12 & 0.10 \\
\cline{1-6}
\multirow[c]{8}{*}{\texttt{Bulgaria}} & 3 & 101.93 & 18000.00 & 91.97 & 104.97 \\
 & 4 & 271.07 & 18000.00 & 124.90 & 148.49 \\
 & 5 & 1440.47 & 18000.00 & 379.42 & 341.45 \\
 & 6 & 3759.51 & 5605.67 & 38.84 & 27.85 \\
 & 7 & 4954.60 & 675.45 & 16.14 & 11.37 \\
 & 8 & 9602.97 & 49.64 & 7.96 & 4.26 \\
 & 9 & 18000.00 & 29.27 & 7.35 & 2.59 \\
 & 10 & 18000.00 & 11.35 & 8.54 & 2.29 \\
\cline{1-6}
\multirow[c]{8}{*}{\texttt{United Kingdom}} & 3 & 3600.65 & 18000.00 & 852.03 & 774.19 \\
 & 4 & 8015.73 & 18000.00 & 938.91 & 787.43 \\
 & 5 & 18000.00 & 18000.00 & 1665.79 & 1498.59 \\
 & 6 & 18000.00 & 18000.00 & 809.84 & 856.50 \\
 & 7 & 18000.00 & 18000.00 & 393.71 & 247.37 \\
 & 8 & 18000.00 & 18000.00 & 1097.35 & 374.60 \\
 & 9 & 18000.00 & 18000.00 & 1615.26 & 522.20 \\
 & 10 & 18000.00 & 18000.00 & 925.93 & 259.56 \\
\cline{1-6}
\multirow[c]{8}{*}{\texttt{Greece}} & 3 & 428.34 & 18000.00 & 393.72 & 315.10 \\
 & 4 & 695.53 & 18000.00 & 1311.11 & 1000.88 \\
 & 5 & 1574.56 & 18000.00 & 309.93 & 164.18 \\
 & 6 & 6985.70 & 18000.00 & 355.58 & 236.43 \\
 & 7 & 18000.00 & 18000.00 & 402.11 & 156.18 \\
 & 8 & 18000.00 & 18000.00 & 1027.81 & 252.06 \\
 & 9 & 18000.00 & 7357.73 & 604.70 & 77.62 \\
 & 10 & 18000.00 & 1202.34 & 360.98 & 70.83 \\
\cline{1-6}
\multirow[c]{8}{*}{\texttt{Germany}} & 3 & 78.37 & 18000.00 & 23.04 & 22.61 \\
 & 4 & 608.38 & 18000.00 & 65.80 & 79.41 \\
 & 5 & 4290.50 & 18000.00 & 121.60 & 134.34 \\
 & 6 & 7684.89 & 18000.00 & 32.52 & 36.79 \\
 & 7 & 18000.00 & 18000.00 & 200.69 & 79.67 \\
 & 8 & 18000.00 & 18000.00 & 109.43 & 57.39 \\
 & 9 & 18000.00 & 18000.01 & 377.07 & 103.51 \\
 & 10 & 18000.00 & 18000.00 & 428.52 & 57.87 \\
\cline{1-6}
\multirow[c]{8}{*}{\texttt{Romania}} & 3 & 18000.00 & 18000.00 & 15003.73 & 18000.00 \\
 & 4 & 18000.00 & 18000.00 & 18000.00 & 13689.03 \\
 & 5 & 18000.00 & 18000.01 & 14632.75 & 3690.48 \\
 & 6 & 18000.00 & 18000.00 & 2639.39 & 1175.32 \\
 & 7 & 18000.00 & 18000.00 & 8809.66 & 722.67 \\
 & 8 & 18000.00 & 18000.00 & 18000.00 & 1219.44 \\
 & 9 & 18000.00 & 18000.01 & 16746.33 & 1090.82 \\
 & 10 & 18000.00 & 18000.00 & 18000.01 & 534.24 \\
\cline{1-6}
\multirow[c]{8}{*}{\texttt{Spain}} & 3 & 1567.02 & 18000.00 & 345.95 & 326.85 \\
 & 4 & 18000.00 & 18000.00 & 1208.84 & 1541.01 \\
 & 5 & 18000.00 & 18000.00 & 1243.61 & 1528.22 \\
 & 6 & 18000.00 & 18000.00 & 7313.83 & 2770.24 \\
 & 7 & 18000.00 & 18000.01 & 16194.48 & 8954.22 \\
 & 8 & 18000.00 & 18000.00 & 18000.00 & 7700.30 \\
 & 9 & 18000.00 & 18000.00 & 1922.54 & 1151.62 \\
 & 10 & 18000.00 & 18000.00 & 2945.97 & 856.33 \\
\cline{1-6}
\bottomrule
\end{tabular}
\end{table}
\begin{table}
\tiny
\caption{Final lower bounds (lb) and upper bounds (ub) on real-world instances (Slight inconsistencies in bounds are due to limited numerical precision)}
\label{tab:results}
\begin{tabular}{|ll||rr|rr|rr|rr|}
\toprule
 &  & \multicolumn{2}{c}{\flow} & \multicolumn{2}{c}{\tree} & \multicolumn{2}{c}{\ERS} & \multicolumn{2}{c}{\ERSTree} \\
 &  &  &  &  &  &  &  &  &  \\
Instance & $k$ &  lb & ub & lb & ub & lb & ub & lb & ub \\
\midrule
\multirow[c]{8}{*}{\texttt{Finland}} & 3 & 116.40 & 116.40 & 116.40 & 116.40 & 116.40 & 116.40 & 116.40 & 116.40 \\
 & 4 & 63.60 & 63.60 & 63.60 & 63.60 & 63.60 & 63.60 & 63.60 & 63.60 \\
 & 5 & 43.50 & 43.50 & 43.50 & 43.50 & 43.50 & 43.50 & 43.50 & 43.50 \\
 & 6 & 28.40 & 28.40 & 28.40 & 28.40 & 28.40 & 28.40 & 28.40 & 28.40 \\
 & 7 & 21.80 & 21.80 & 21.80 & 21.80 & 21.80 & 21.80 & 21.80 & 21.80 \\
 & 8 & 15.50 & 15.50 & 15.50 & 15.50 & 15.50 & 15.50 & 15.50 & 15.50 \\
 & 9 & 11.30 & 11.30 & 11.30 & 11.30 & 11.30 & 11.30 & 11.30 & 11.30 \\
 & 10 & 7.70 & 7.70 & 7.70 & 7.70 & 7.70 & 7.70 & 7.70 & 7.70 \\
\cline{1-10}
\multirow[c]{8}{*}{\texttt{Bulgaria}} & 3 & 588.90 & 588.90 & 287.80 & 588.90 & 588.90 & 588.90 & 588.90 & 588.90 \\
 & 4 & 388.40 & 388.40 & 262.80 & 388.40 & 388.40 & 388.40 & 388.40 & 388.40 \\
 & 5 & 295.70 & 295.70 & 214.50 & 297.20 & 295.70 & 295.70 & 295.70 & 295.70 \\
 & 6 & 202.80 & 202.80 & 202.80 & 202.80 & 202.80 & 202.80 & 202.80 & 202.80 \\
 & 7 & 149.90 & 149.90 & 149.90 & 149.90 & 149.90 & 149.90 & 149.90 & 149.90 \\
 & 8 & 114.90 & 114.90 & 114.90 & 114.90 & 114.90 & 114.90 & 114.90 & 114.90 \\
 & 9 & 70.70 & 96.00 & 96.00 & 96.00 & 96.00 & 96.00 & 96.00 & 96.00 \\
 & 10 & 59.10 & 78.60 & 78.60 & 78.60 & 78.60 & 78.60 & 78.60 & 78.60 \\
\cline{1-10}
\multirow[c]{8}{*}{\texttt{United Kingdom}} & 3 & 348.00 & 348.00 & 92.50 & 386.90 & 348.00 & 348.00 & 348.00 & 348.00 \\
 & 4 & 245.50 & 245.50 & 80.70 & 245.50 & 245.50 & 245.50 & 245.50 & 245.50 \\
 & 5 & 157.20 & 182.50 & 81.10 & 182.50 & 182.50 & 182.50 & 182.50 & 182.50 \\
 & 6 & 100.60 & 137.90 & 59.10 & 137.90 & 137.90 & 137.90 & 137.90 & 137.90 \\
 & 7 & 54.50 & 107.80 & 62.90 & 107.70 & 107.70 & 107.70 & 107.70 & 107.70 \\
 & 8 & 34.20 & 89.90 & 60.40 & 89.90 & 89.90 & 89.90 & 89.90 & 89.90 \\
 & 9 & 25.10 & 75.70 & 55.20 & 75.60 & 75.60 & 75.60 & 75.60 & 75.60 \\
 & 10 & 19.70 & 62.20 & 50.10 & 62.00 & 62.00 & 62.00 & 62.00 & 62.00 \\
\cline{1-10}
\multirow[c]{8}{*}{\texttt{Greece}} & 3 & 979.80 & 979.80 & 376.70 & 979.80 & 979.80 & 979.80 & 979.80 & 979.80 \\
 & 4 & 716.90 & 716.90 & 360.60 & 716.90 & 716.90 & 716.90 & 716.90 & 716.90 \\
 & 5 & 500.50 & 500.50 & 305.00 & 500.50 & 500.50 & 500.50 & 500.50 & 500.50 \\
 & 6 & 395.70 & 395.70 & 269.30 & 395.70 & 395.70 & 395.70 & 395.70 & 395.70 \\
 & 7 & 230.20 & 316.50 & 265.10 & 316.50 & 316.50 & 316.50 & 316.50 & 316.50 \\
 & 8 & 105.20 & 271.50 & 235.90 & 265.60 & 265.60 & 265.60 & 265.60 & 265.60 \\
 & 9 & 54.00 & 220.60 & 220.60 & 220.60 & 220.60 & 220.60 & 220.60 & 220.60 \\
 & 10 & 54.10 & 188.10 & 186.30 & 186.30 & 186.30 & 186.30 & 186.30 & 186.30 \\
\cline{1-10}
\multirow[c]{8}{*}{\texttt{Germany}} & 3 & 291.20 & 291.20 & 59.70 & 406.20 & 291.20 & 291.20 & 291.20 & 291.20 \\
 & 4 & 187.50 & 187.50 & 57.20 & 187.50 & 187.50 & 187.50 & 187.50 & 187.50 \\
 & 5 & 135.90 & 135.90 & 49.70 & 142.70 & 135.90 & 135.90 & 135.90 & 135.90 \\
 & 6 & 91.10 & 91.10 & 50.10 & 91.10 & 91.10 & 91.10 & 91.10 & 91.10 \\
 & 7 & 57.40 & 72.40 & 44.40 & 72.40 & 72.40 & 72.40 & 72.40 & 72.40 \\
 & 8 & 36.00 & 59.00 & 44.50 & 59.00 & 59.00 & 59.00 & 59.00 & 59.00 \\
 & 9 & 26.90 & 50.50 & 35.70 & 50.50 & 50.50 & 50.50 & 50.50 & 50.50 \\
 & 10 & 20.20 & 43.50 & 39.40 & 43.50 & 43.50 & 43.50 & 43.50 & 43.50 \\
\cline{1-10}
\multirow[c]{8}{*}{\texttt{Romania}} & 3 & 594.70 & 766.40 & 153.00 & $\infty$ & 764.80 & 764.80 & 705.40 & 764.80 \\
 & 4 & 259.90 & 514.30 & 159.60 & 631.80 & 488.70 & 514.30 & 514.30 & 514.30 \\
 & 5 & 161.50 & 377.70 & 135.80 & 453.10 & 377.70 & 377.70 & 377.70 & 377.70 \\
 & 6 & 96.50 & 290.20 & 148.20 & 357.80 & 290.20 & 290.20 & 290.20 & 290.20 \\
 & 7 & 61.70 & 238.10 & 120.40 & 238.10 & 235.80 & 235.80 & 235.80 & 235.80 \\
 & 8 & 37.80 & 202.70 & 127.30 & 201.50 & 199.50 & 201.50 & 201.50 & 201.50 \\
 & 9 & 23.50 & 172.40 & 118.50 & 172.40 & 172.40 & 172.40 & 172.40 & 172.40 \\
 & 10 & 23.40 & 151.00 & 113.80 & 148.10 & 145.10 & 148.10 & 148.10 & 148.10 \\
\cline{1-10}
\multirow[c]{8}{*}{\texttt{Spain}} & 3 & 1191.80 & 1191.80 & 182.10 & $\infty$ & 1191.80 & 1191.80 & 1191.80 & 1191.80 \\
 & 4 & 587.70 & 794.10 & 170.40 & $\infty$ & 777.80 & 777.80 & 777.80 & 777.80 \\
 & 5 & 313.20 & 559.60 & 159.00 & 651.00 & 559.60 & 559.60 & 559.60 & 559.60 \\
 & 6 & 188.00 & 442.40 & 148.80 & 541.80 & 442.40 & 442.40 & 442.40 & 442.40 \\
 & 7 & 101.00 & 372.20 & 144.10 & 417.50 & 359.60 & 359.60 & 359.60 & 359.60 \\
 & 8 & 71.40 & 303.20 & 129.50 & 327.70 & 287.30 & 295.40 & 295.40 & 295.40 \\
 & 9 & 59.30 & 235.60 & 123.90 & 235.60 & 235.60 & 235.60 & 235.60 & 235.60 \\
 & 10 & 36.10 & 206.70 & 119.20 & 200.70 & 200.70 & 200.70 & 200.70 & 200.70 \\
\cline{1-10}
\bottomrule
\end{tabular}
\end{table}

\begin{table}
\caption{Detailed runtime results on grid instances}
\tiny
\label{tab:results}
\begin{tabular}{|lll||r|r|r|r|}
\toprule
 &  & & \flow & \tree & \ERS & \ERSTree \\
Size & $k$ & $\rho$ &  &  &  &  \\
\midrule
\multirow[c]{32}{*}{$4 \times 4$} & \multirow[c]{4}{*}{3} & 0 & 1.49 & 1.17 & 0.89 & 0.61 \\
 &  & 30 & 0.91 & 0.45 & 0.41 & 0.12 \\
 &  & 60 & 1.44 & 0.98 & 0.68 & 0.37 \\
 &  & 90 & 0.65 & 1.18 & 0.36 & 0.39 \\
\cline{2-7}
 & \multirow[c]{4}{*}{4} & 0 & 3.86 & 0.34 & 0.53 & 0.25 \\
 &  & 30 & 1.01 & 0.11 & 0.14 & 0.03 \\
 &  & 60 & 3.47 & 0.57 & 0.39 & 0.24 \\
 &  & 90 & 3.56 & 0.31 & 0.37 & 0.52 \\
\cline{2-7}
 & \multirow[c]{4}{*}{5} & 0 & 2.05 & 0.10 & 0.15 & 0.07 \\
 &  & 30 & 1.02 & 0.32 & 0.07 & 0.02 \\
 &  & 60 & 3.82 & 0.18 & 0.13 & 0.08 \\
 &  & 90 & 1.62 & 0.16 & 0.16 & 0.13 \\
\cline{2-7}
 & \multirow[c]{4}{*}{6} & 0 & 2.42 & 0.09 & 0.14 & 0.07 \\
 &  & 30 & 1.26 & 0.15 & 0.04 & 0.02 \\
 &  & 60 & 3.32 & 0.09 & 0.09 & 0.05 \\
 &  & 90 & 1.98 & 0.07 & 0.02 & 0.03 \\
\cline{2-7}
 & \multirow[c]{4}{*}{7} & 0 & 1.25 & 0.09 & 0.02 & 0.02 \\
 &  & 30 & 0.91 & 0.07 & 0.07 & 0.02 \\
 &  & 60 & 2.08 & 0.09 & 0.03 & 0.03 \\
 &  & 90 & 0.96 & 0.05 & 0.01 & 0.01 \\
\cline{2-7}
 & \multirow[c]{4}{*}{8} & 0 & 1.13 & 0.03 & 0.01 & 0.02 \\
 &  & 30 & 0.70 & 0.03 & 0.02 & 0.01 \\
 &  & 60 & 1.59 & 0.08 & 0.01 & 0.01 \\
 &  & 90 & 1.26 & 0.06 & 0.01 & 0.01 \\
\cline{2-7}
 & \multirow[c]{4}{*}{9} & 0 & 0.74 & 0.03 & 0.02 & 0.02 \\
 &  & 30 & 0.58 & 0.04 & 0.03 & 0.01 \\
 &  & 60 & 0.79 & 0.03 & 0.01 & 0.01 \\
 &  & 90 & 0.77 & 0.07 & 0.01 & 0.02 \\
\cline{2-7}
 & \multirow[c]{4}{*}{10} & 0 & 0.62 & 0.02 & 0.01 & 0.02 \\
 &  & 30 & 0.84 & 0.02 & 0.01 & 0.01 \\
 &  & 60 & 0.59 & 0.01 & 0.01 & 0.01 \\
 &  & 90 & 0.77 & 0.01 & 0.01 & 0.01 \\
\cline{1-7} \cline{2-7}
\multirow[c]{32}{*}{$5 \times 5$} & \multirow[c]{4}{*}{3} & 0 & 54.45 & 561.41 & 5.20 & 5.05 \\
 &  & 30 & 182.30 & 385.23 & 11.96 & 12.47 \\
 &  & 60 & 2.93 & 173.43 & 0.62 & 0.70 \\
 &  & 90 & 34.89 & 472.19 & 6.70 & 6.46 \\
\cline{2-7}
 & \multirow[c]{4}{*}{4} & 0 & 153.65 & 220.07 & 7.15 & 7.77 \\
 &  & 30 & 71.01 & 11.01 & 0.96 & 0.67 \\
 &  & 60 & 12.42 & 13.53 & 1.46 & 0.55 \\
 &  & 90 & 88.26 & 112.57 & 3.97 & 3.63 \\
\cline{2-7}
 & \multirow[c]{4}{*}{5} & 0 & 555.06 & 44.08 & 2.49 & 3.46 \\
 &  & 30 & 208.24 & 7.74 & 3.70 & 2.08 \\
 &  & 60 & 80.52 & 1.90 & 1.55 & 0.98 \\
 &  & 90 & 189.00 & 25.28 & 5.19 & 4.03 \\
\cline{2-7}
 & \multirow[c]{4}{*}{6} & 0 & 1158.00 & 11.95 & 2.49 & 2.38 \\
 &  & 30 & 787.97 & 2.83 & 2.12 & 0.92 \\
 &  & 60 & 210.69 & 0.95 & 2.27 & 0.68 \\
 &  & 90 & 301.77 & 7.15 & 2.45 & 2.55 \\
\cline{2-7}
 & \multirow[c]{4}{*}{7} & 0 & 1770.82 & 2.80 & 0.50 & 0.52 \\
 &  & 30 & 774.30 & 0.68 & 0.75 & 0.26 \\
 &  & 60 & 374.82 & 0.89 & 0.88 & 0.37 \\
 &  & 90 & 373.09 & 1.21 & 1.11 & 0.21 \\
\cline{2-7}
 & \multirow[c]{4}{*}{8} & 0 & 5727.13 & 0.88 & 0.35 & 0.63 \\
 &  & 30 & 1508.63 & 0.65 & 0.96 & 0.18 \\
 &  & 60 & 513.96 & 0.53 & 0.87 & 0.26 \\
 &  & 90 & 372.40 & 0.35 & 0.35 & 0.15 \\
\cline{2-7}
 & \multirow[c]{4}{*}{9} & 0 & 2421.16 & 0.44 & 0.12 & 0.20 \\
 &  & 30 & 1631.03 & 0.41 & 0.89 & 0.15 \\
 &  & 60 & 588.59 & 0.32 & 0.68 & 0.18 \\
 &  & 90 & 332.21 & 0.14 & 0.05 & 0.07 \\
\cline{2-7}
 & \multirow[c]{4}{*}{10} & 0 & 647.25 & 0.29 & 0.09 & 0.11 \\
 &  & 30 & 1851.12 & 0.26 & 0.31 & 0.12 \\
 &  & 60 & 635.96 & 0.64 & 0.81 & 0.11 \\
 &  & 90 & 364.20 & 0.09 & 0.05 & 0.06 \\
\cline{1-7} \cline{2-7}
\end{tabular}
\end{table}
\begin{table}
\caption{Detailed runtime results on grid instances (continuation)}
\tiny
\label{tab:results}
\begin{tabular}{|lll||r|r|r|r|}
\toprule
 &  & & \flow & \tree & \ERS & \ERSTree \\
Size & $k$ & $\rho$ &  &  &  &  \\
\midrule
\multirow[c]{32}{*}{$6 \times 6$} & \multirow[c]{4}{*}{3} & 0 & 3173.17 & 18000.00 & 1933.51 & 2257.07 \\
 &  & 30 & 3538.70 & 18000.00 & 1011.37 & 2954.76 \\
 &  & 60 & 1475.93 & 18000.00 & 204.30 & 150.30 \\
 &  & 90 & 425.56 & 18000.00 & 481.72 & 1071.71 \\
\cline{2-7}
 & \multirow[c]{4}{*}{4} & 0 & 18000.00 & 18000.00 & 3621.91 & 2528.45 \\
 &  & 30 & 4745.47 & 18000.00 & 48.39 & 23.00 \\
 &  & 60 & 18000.00 & 18000.00 & 413.37 & 684.95 \\
 &  & 90 & 6180.58 & 18000.00 & 1282.13 & 547.69 \\
\cline{2-7}
 & \multirow[c]{4}{*}{5} & 0 & 18000.00 & 18000.00 & 409.80 & 305.68 \\
 &  & 30 & 18000.00 & 18000.00 & 70.45 & 33.16 \\
 &  & 60 & 18000.00 & 18000.00 & 657.07 & 1018.77 \\
 &  & 90 & 14059.45 & 18000.00 & 37.85 & 24.30 \\
\cline{2-7}
 & \multirow[c]{4}{*}{6} & 0 & 18000.00 & 18000.00 & 113.02 & 32.94 \\
 &  & 30 & 18000.00 & 13811.24 & 162.52 & 91.77 \\
 &  & 60 & 18000.00 & 4764.24 & 49.50 & 30.16 \\
 &  & 90 & 18000.00 & 1215.02 & 21.49 & 5.40 \\
\cline{2-7}
 & \multirow[c]{4}{*}{7} & 0 & 18000.00 & 4461.63 & 102.75 & 23.46 \\
 &  & 30 & 18000.00 & 2878.29 & 163.11 & 28.94 \\
 &  & 60 & 18000.00 & 145.28 & 17.52 & 7.50 \\
 &  & 90 & 18000.00 & 93.87 & 20.80 & 4.34 \\
\cline{2-7}
 & \multirow[c]{4}{*}{8} & 0 & 18000.00 & 539.19 & 155.50 & 27.37 \\
 &  & 30 & 18000.00 & 366.91 & 158.54 & 14.83 \\
 &  & 60 & 18000.00 & 46.86 & 16.82 & 4.85 \\
 &  & 90 & 18000.00 & 29.57 & 10.34 & 4.57 \\
\cline{2-7}
 & \multirow[c]{4}{*}{9} & 0 & 18000.00 & 252.29 & 85.21 & 23.07 \\
 &  & 30 & 18000.00 & 36.78 & 76.15 & 6.48 \\
 &  & 60 & 18000.00 & 12.53 & 11.92 & 3.48 \\
 &  & 90 & 18000.00 & 10.00 & 5.83 & 1.35 \\
\cline{2-7}
 & \multirow[c]{4}{*}{10} & 0 & 18000.01 & 201.91 & 11.44 & 3.92 \\
 &  & 30 & 18000.00 & 12.67 & 37.79 & 2.53 \\
 &  & 60 & 18000.00 & 7.98 & 3.59 & 1.57 \\
 &  & 90 & 18000.00 & 6.16 & 3.21 & 0.57 \\
\cline{1-7} \cline{2-7}
\multirow[c]{32}{*}{$7 \times 7$} & \multirow[c]{4}{*}{3} & 0 & 18000.00 & 18000.00 & 18000.00 & 18000.04 \\
 &  & 30 & 18000.00 & 18000.00 & 18000.00 & 18000.00 \\
 &  & 60 & 18000.00 & 18000.01 & 7748.44 & 15348.45 \\
 &  & 90 & 18000.00 & 18000.00 & 13809.59 & 7894.15 \\
\cline{2-7}
 & \multirow[c]{4}{*}{4} & 0 & 18000.00 & 18000.00 & 18000.00 & 18000.00 \\
 &  & 30 & 18000.00 & 18000.00 & 18000.06 & 18000.08 \\
 &  & 60 & 18000.01 & 18000.01 & 8409.97 & 7549.17 \\
 &  & 90 & 18000.00 & 18000.00 & 18000.00 & 14821.58 \\
\cline{2-7}
 & \multirow[c]{4}{*}{5} & 0 & 18000.00 & 18000.00 & 18000.00 & 15188.95 \\
 &  & 30 & 18000.00 & 18000.00 & 18000.00 & 18000.04 \\
 &  & 60 & 18000.00 & 18000.00 & 1642.46 & 1591.97 \\
 &  & 90 & 18000.00 & 18000.00 & 18000.00 & 18000.00 \\
\cline{2-7}
 & \multirow[c]{4}{*}{6} & 0 & 18000.00 & 18000.00 & 7141.67 & 1682.89 \\
 &  & 30 & 18000.00 & 18000.00 & 18000.00 & 18000.00 \\
 &  & 60 & 18000.00 & 18000.00 & 18000.01 & 3438.08 \\
 &  & 90 & 18000.00 & 18000.00 & 18000.00 & 11095.24 \\
\cline{2-7}
 & \multirow[c]{4}{*}{7} & 0 & 18000.00 & 18000.00 & 18000.05 & 7400.82 \\
 &  & 30 & 18000.00 & 18000.00 & 18000.04 & 18000.00 \\
 &  & 60 & 18000.00 & 18000.00 & 18000.01 & 4318.91 \\
 &  & 90 & 18000.00 & 18000.01 & 7958.19 & 2905.39 \\
\cline{2-7}
 & \multirow[c]{4}{*}{8} & 0 & 18000.00 & 18000.00 & 18000.01 & 5332.83 \\
 &  & 30 & 18000.00 & 18000.00 & 18000.00 & 11435.34 \\
 &  & 60 & 18000.00 & 18000.00 & 11379.43 & 830.86 \\
 &  & 90 & 18000.00 & 18000.01 & 6033.23 & 1408.61 \\
\cline{2-7}
 & \multirow[c]{4}{*}{9} & 0 & 18000.00 & 18000.00 & 3607.06 & 447.97 \\
 &  & 30 & 18000.00 & 18000.00 & 2385.48 & 585.35 \\
 &  & 60 & 18000.00 & 18000.00 & 12975.71 & 859.22 \\
 &  & 90 & 18000.01 & 18000.00 & 1603.00 & 327.62 \\
\cline{2-7}
 & \multirow[c]{4}{*}{10} & 0 & 18000.00 & 18000.00 & 3614.20 & 385.03 \\
 &  & 30 & 18000.00 & 18000.00 & 2915.21 & 249.53 \\
 &  & 60 & 18000.00 & 18000.00 & 11413.25 & 323.36 \\
 &  & 90 & 18000.00 & 18000.00 & 691.61 & 89.37 \\
\cline{1-7} \cline{2-7}
\end{tabular}
\end{table}
\begin{table}
\caption{Detailed runtime results on grid instances (continuation)}
\tiny
\label{tab:results}
\begin{tabular}{|lll||r|r|r|r|}
\toprule
 &  & & \flow & \tree & \ERS & \ERSTree \\
Size & $k$ & $\rho$ &  &  &  &  \\
\midrule
\multirow[c]{32}{*}{$8 \times 8$} & \multirow[c]{4}{*}{3} & 0 & 18000.00 & 18000.01 & 18000.00 & 18000.01 \\
 &  & 30 & 18000.00 & 18000.00 & 18000.01 & 18000.01 \\
 &  & 60 & 18000.00 & 18000.00 & 18000.00 & 18000.01 \\
 &  & 90 & 18000.00 & 18000.01 & 18000.00 & 18000.00 \\
\cline{2-7}
 & \multirow[c]{4}{*}{4} & 0 & 18000.02 & 18000.00 & 18000.00 & 18000.00 \\
 &  & 30 & 18000.00 & 18000.02 & 18000.03 & 18000.00 \\
 &  & 60 & 18000.00 & 18000.00 & 18000.00 & 18000.00 \\
 &  & 90 & 18000.00 & 18000.01 & 18000.05 & 18000.00 \\
\cline{2-7}
 & \multirow[c]{4}{*}{5} & 0 & 18000.00 & 18000.00 & 18000.00 & 18000.18 \\
 &  & 30 & 18000.00 & 18000.00 & 18000.01 & 18000.00 \\
 &  & 60 & 18000.00 & 18000.00 & 18000.01 & 18000.01 \\
 &  & 90 & 18000.00 & 18000.00 & 18000.00 & 18000.00 \\
\cline{2-7}
 & \multirow[c]{4}{*}{6} & 0 & 18000.00 & 18000.00 & 18000.01 & 18000.01 \\
 &  & 30 & 18000.00 & 18000.00 & 18000.00 & 18000.00 \\
 &  & 60 & 18000.00 & 18000.00 & 18000.01 & 18000.01 \\
 &  & 90 & 18000.00 & 18000.00 & 18000.01 & 18000.01 \\
\cline{2-7}
 & \multirow[c]{4}{*}{7} & 0 & 18000.00 & 18000.00 & 18000.00 & 18000.00 \\
 &  & 30 & 18000.01 & 18000.00 & 18000.01 & 18000.01 \\
 &  & 60 & 18000.00 & 18000.00 & 18000.01 & 18000.00 \\
 &  & 90 & 18000.00 & 18000.01 & 18000.00 & 18000.01 \\
\cline{2-7}
 & \multirow[c]{4}{*}{8} & 0 & 18000.00 & 18000.00 & 18000.01 & 18000.00 \\
 &  & 30 & 18000.00 & 18000.00 & 18000.01 & 18000.00 \\
 &  & 60 & 18000.01 & 18000.00 & 18000.00 & 18000.00 \\
 &  & 90 & 18000.01 & 18000.00 & 18000.01 & 18000.00 \\
\cline{2-7}
 & \multirow[c]{4}{*}{9} & 0 & 18000.00 & 18000.00 & 18000.00 & 18000.00 \\
 &  & 30 & 18000.00 & 18000.05 & 18000.01 & 12743.91 \\
 &  & 60 & 18000.00 & 18000.00 & 18000.01 & 18000.00 \\
 &  & 90 & 18000.01 & 18000.00 & 18000.00 & 18000.01 \\
\cline{2-7}
 & \multirow[c]{4}{*}{10} & 0 & 18000.01 & 18000.00 & 18000.00 & 18000.00 \\
 &  & 30 & 18000.01 & 18000.01 & 18000.11 & 18000.01 \\
 &  & 60 & 18000.00 & 18000.01 & 18000.00 & 18000.01 \\
 &  & 90 & 18000.01 & 18000.00 & 18000.01 & 18000.00 \\
\cline{1-7} \cline{2-7}
\multirow[c]{32}{*}{$9 \times 9$} & \multirow[c]{4}{*}{3} & 0 & 18000.00 & 18000.00 & 18000.01 & 18000.01 \\
 &  & 30 & 18000.00 & 18000.00 & 18000.01 & 18000.01 \\
 &  & 60 & 18000.10 & 18000.00 & 18000.01 & 18000.01 \\
 &  & 90 & 18000.13 & 18000.00 & 18000.01 & 18000.01 \\
\cline{2-7}
 & \multirow[c]{4}{*}{4} & 0 & 18000.95 & 18000.00 & 18000.01 & 18000.01 \\
 &  & 30 & 18000.01 & 18000.00 & 18000.01 & 18000.01 \\
 &  & 60 & 18000.00 & 18000.00 & 18000.01 & 18000.01 \\
 &  & 90 & 18000.00 & 18000.00 & 18000.01 & 18000.01 \\
\cline{2-7}
 & \multirow[c]{4}{*}{5} & 0 & 18000.01 & 18000.00 & 18000.01 & 18000.01 \\
 &  & 30 & 18000.01 & 18000.00 & 18000.01 & 18000.01 \\
 &  & 60 & 18000.01 & 18000.02 & 18000.01 & 18000.01 \\
 &  & 90 & 18000.00 & 18000.00 & 18000.01 & 18000.01 \\
\cline{2-7}
 & \multirow[c]{4}{*}{6} & 0 & 18000.01 & 18000.00 & 18000.01 & 18000.01 \\
 &  & 30 & 18000.01 & 18000.00 & 18000.01 & 18000.01 \\
 &  & 60 & 18000.01 & 18000.00 & 18000.01 & 18000.00 \\
 &  & 90 & 18000.01 & 18000.00 & 18000.12 & 18000.01 \\
\cline{2-7}
 & \multirow[c]{4}{*}{7} & 0 & 18000.01 & 18000.00 & 18000.00 & 18000.01 \\
 &  & 30 & 18000.01 & 18000.02 & 18000.01 & 18000.01 \\
 &  & 60 & 18000.01 & 18000.03 & 18000.01 & 18000.07 \\
 &  & 90 & 18000.01 & 18000.03 & 18000.01 & 18000.01 \\
\cline{2-7}
 & \multirow[c]{4}{*}{8} & 0 & 18000.01 & 18000.00 & 18000.01 & 18000.01 \\
 &  & 30 & 18000.01 & 18000.00 & 18000.01 & 18000.03 \\
 &  & 60 & 18000.01 & 18000.02 & 18000.01 & 18000.01 \\
 &  & 90 & 18000.01 & 18000.01 & 18000.01 & 18000.01 \\
\cline{2-7}
 & \multirow[c]{4}{*}{9} & 0 & 18000.01 & 18000.00 & 18000.01 & 18000.01 \\
 &  & 30 & 18000.01 & 18000.01 & 18000.01 & 18000.01 \\
 &  & 60 & 18000.01 & 18000.01 & 18000.01 & 18000.01 \\
 &  & 90 & 18000.01 & 18000.01 & 18000.01 & 18000.01 \\
\cline{2-7}
 & \multirow[c]{4}{*}{10} & 0 & 18000.01 & 18000.00 & 18000.01 & 18000.02 \\
 &  & 30 & 18000.01 & 18000.01 & 18000.01 & 18000.01 \\
 &  & 60 & 18000.01 & 18000.01 & 18000.01 & 18000.01 \\
 &  & 90 & 18000.01 & 18000.03 & 18000.01 & 18000.01 \\
\cline{1-7} \cline{2-7}
\bottomrule
\end{tabular}
\end{table}

\begin{table}
\tiny
\caption{Final lower bounds (lb) and upper bounds (ub) on grid instances (Slight inconsistencies in bounds are due to limited numerical precision)}
\label{tab:results}
\begin{tabular}{lllrrrrrrrr}
\toprule
 &  &  & \multicolumn{2}{c}{\flow} & \multicolumn{2}{c}{\tree} & \multicolumn{2}{c}{\ERS} & \multicolumn{2}{c}{\ERSTree} \\
 &  &  & lb & ub & lb & ub & lb & ub & lb & ub \\
Size & $k$ & $\rho$ &  &  &  &  &  &  &  &  \\
\midrule
\multirow[c]{32}{*}{$4 \times 4$} & \multirow[c]{4}{*}{3} & 0 & 248.43 & 248.43 & 248.43 & 248.43 & 248.43 & 248.43 & 248.43 & 248.43 \\
 &  & 30 & 229.24 & 229.24 & 229.24 & 229.24 & 229.24 & 229.24 & 229.23 & 229.24 \\
 &  & 60 & 229.79 & 229.79 & 229.79 & 229.79 & 229.79 & 229.79 & 229.79 & 229.79 \\
 &  & 90 & 314.47 & 314.47 & 314.47 & 314.47 & 314.47 & 314.47 & 314.44 & 314.47 \\
\cline{2-11}
 & \multirow[c]{4}{*}{4} & 0 & 156.00 & 156.00 & 156.00 & 156.00 & 156.00 & 156.00 & 156.00 & 156.00 \\
 &  & 30 & 123.71 & 123.72 & 123.72 & 123.72 & 123.72 & 123.72 & 123.72 & 123.72 \\
 &  & 60 & 150.13 & 150.13 & 150.13 & 150.13 & 150.13 & 150.13 & 150.13 & 150.13 \\
 &  & 90 & 212.42 & 212.42 & 212.42 & 212.42 & 212.42 & 212.42 & 212.42 & 212.42 \\
\cline{2-11}
 & \multirow[c]{4}{*}{5} & 0 & 92.54 & 92.54 & 92.54 & 92.54 & 92.54 & 92.54 & 92.54 & 92.54 \\
 &  & 30 & 78.03 & 78.03 & 78.03 & 78.03 & 78.03 & 78.03 & 78.03 & 78.03 \\
 &  & 60 & 103.09 & 103.09 & 103.09 & 103.09 & 103.09 & 103.09 & 103.09 & 103.09 \\
 &  & 90 & 141.42 & 141.42 & 141.42 & 141.42 & 141.42 & 141.42 & 141.42 & 141.42 \\
\cline{2-11}
 & \multirow[c]{4}{*}{6} & 0 & 66.22 & 66.22 & 66.22 & 66.22 & 66.22 & 66.22 & 66.22 & 66.22 \\
 &  & 30 & 52.22 & 52.22 & 52.22 & 52.22 & 52.22 & 52.22 & 52.22 & 52.22 \\
 &  & 60 & 73.82 & 73.82 & 73.82 & 73.82 & 73.82 & 73.82 & 73.82 & 73.82 \\
 &  & 90 & 85.45 & 85.45 & 85.45 & 85.45 & 85.45 & 85.45 & 85.45 & 85.45 \\
\cline{2-11}
 & \multirow[c]{4}{*}{7} & 0 & 42.52 & 42.52 & 42.52 & 42.52 & 42.52 & 42.52 & 42.52 & 42.52 \\
 &  & 30 & 37.36 & 37.36 & 37.36 & 37.36 & 37.36 & 37.36 & 37.36 & 37.36 \\
 &  & 60 & 48.48 & 48.48 & 48.48 & 48.48 & 48.48 & 48.48 & 48.48 & 48.48 \\
 &  & 90 & 56.51 & 56.51 & 56.51 & 56.51 & 56.51 & 56.51 & 56.51 & 56.51 \\
\cline{2-11}
 & \multirow[c]{4}{*}{8} & 0 & 30.74 & 30.74 & 30.74 & 30.74 & 30.74 & 30.74 & 30.74 & 30.74 \\
 &  & 30 & 22.96 & 22.96 & 22.96 & 22.96 & 22.96 & 22.96 & 22.96 & 22.96 \\
 &  & 60 & 29.10 & 29.10 & 29.10 & 29.10 & 29.10 & 29.10 & 29.10 & 29.10 \\
 &  & 90 & 44.21 & 44.21 & 44.21 & 44.21 & 44.21 & 44.21 & 44.21 & 44.21 \\
\cline{2-11}
 & \multirow[c]{4}{*}{9} & 0 & 22.58 & 22.58 & 22.58 & 22.58 & 22.58 & 22.58 & 22.58 & 22.58 \\
 &  & 30 & 16.75 & 16.75 & 16.75 & 16.75 & 16.75 & 16.75 & 16.75 & 16.75 \\
 &  & 60 & 13.46 & 13.46 & 13.46 & 13.46 & 13.46 & 13.46 & 13.46 & 13.46 \\
 &  & 90 & 28.22 & 28.22 & 28.22 & 28.22 & 28.22 & 28.22 & 28.22 & 28.22 \\
\cline{2-11}
 & \multirow[c]{4}{*}{10} & 0 & 15.63 & 15.63 & 15.63 & 15.63 & 15.63 & 15.63 & 15.63 & 15.63 \\
 &  & 30 & 10.66 & 10.66 & 10.66 & 10.66 & 10.66 & 10.66 & 10.66 & 10.66 \\
 &  & 60 & 5.82 & 5.82 & 5.82 & 5.82 & 5.82 & 5.82 & 5.82 & 5.82 \\
 &  & 90 & 15.92 & 15.92 & 15.92 & 15.92 & 15.92 & 15.92 & 15.92 & 15.92 \\
\cline{1-11} \cline{2-11}
\multirow[c]{32}{*}{$5 \times 5$} & \multirow[c]{4}{*}{3} & 0 & 619.28 & 619.28 & 619.22 & 619.28 & 619.28 & 619.28 & 619.28 & 619.28 \\
 &  & 30 & 644.22 & 644.22 & 644.20 & 644.22 & 644.22 & 644.22 & 644.22 & 644.22 \\
 &  & 60 & 479.24 & 479.24 & 479.20 & 479.24 & 479.24 & 479.24 & 479.24 & 479.24 \\
 &  & 90 & 860.48 & 860.48 & 860.48 & 860.48 & 860.48 & 860.48 & 860.48 & 860.48 \\
\cline{2-11}
 & \multirow[c]{4}{*}{4} & 0 & 451.08 & 451.08 & 451.07 & 451.08 & 451.08 & 451.08 & 451.08 & 451.08 \\
 &  & 30 & 330.14 & 330.14 & 330.14 & 330.14 & 330.14 & 330.14 & 330.12 & 330.14 \\
 &  & 60 & 305.93 & 305.93 & 305.93 & 305.93 & 305.93 & 305.93 & 305.91 & 305.93 \\
 &  & 90 & 549.19 & 549.19 & 549.16 & 549.19 & 549.19 & 549.19 & 549.19 & 549.19 \\
\cline{2-11}
 & \multirow[c]{4}{*}{5} & 0 & 311.59 & 311.59 & 311.57 & 311.59 & 311.59 & 311.59 & 311.59 & 311.59 \\
 &  & 30 & 265.15 & 265.15 & 265.15 & 265.15 & 265.15 & 265.15 & 265.15 & 265.15 \\
 &  & 60 & 236.02 & 236.02 & 236.02 & 236.02 & 236.02 & 236.02 & 236.02 & 236.02 \\
 &  & 90 & 393.64 & 393.64 & 393.64 & 393.64 & 393.64 & 393.64 & 393.64 & 393.64 \\
\cline{2-11}
 & \multirow[c]{4}{*}{6} & 0 & 247.87 & 247.89 & 247.89 & 247.89 & 247.89 & 247.89 & 247.89 & 247.89 \\
 &  & 30 & 201.96 & 201.96 & 201.96 & 201.96 & 201.96 & 201.96 & 201.96 & 201.96 \\
 &  & 60 & 188.43 & 188.43 & 188.43 & 188.43 & 188.43 & 188.43 & 188.43 & 188.43 \\
 &  & 90 & 270.26 & 270.26 & 270.26 & 270.26 & 270.26 & 270.26 & 270.26 & 270.26 \\
\cline{2-11}
 & \multirow[c]{4}{*}{7} & 0 & 191.14 & 191.14 & 191.14 & 191.14 & 191.14 & 191.14 & 191.14 & 191.14 \\
 &  & 30 & 151.30 & 151.30 & 151.30 & 151.30 & 151.30 & 151.30 & 151.30 & 151.30 \\
 &  & 60 & 150.96 & 150.96 & 150.96 & 150.96 & 150.96 & 150.96 & 150.96 & 150.96 \\
 &  & 90 & 184.51 & 184.52 & 184.52 & 184.52 & 184.52 & 184.52 & 184.52 & 184.52 \\
\cline{2-11}
 & \multirow[c]{4}{*}{8} & 0 & 151.05 & 151.06 & 151.06 & 151.06 & 151.06 & 151.06 & 151.06 & 151.06 \\
 &  & 30 & 123.56 & 123.57 & 123.57 & 123.57 & 123.57 & 123.57 & 123.57 & 123.57 \\
 &  & 60 & 121.36 & 121.36 & 121.36 & 121.36 & 121.36 & 121.36 & 121.36 & 121.36 \\
 &  & 90 & 140.85 & 140.85 & 140.85 & 140.85 & 140.85 & 140.85 & 140.85 & 140.85 \\
\cline{2-11}
 & \multirow[c]{4}{*}{9} & 0 & 113.70 & 113.71 & 113.71 & 113.71 & 113.71 & 113.71 & 113.71 & 113.71 \\
 &  & 30 & 101.35 & 101.36 & 101.36 & 101.36 & 101.36 & 101.36 & 101.36 & 101.36 \\
 &  & 60 & 100.74 & 100.75 & 100.75 & 100.75 & 100.75 & 100.75 & 100.75 & 100.75 \\
 &  & 90 & 109.83 & 109.83 & 109.83 & 109.83 & 109.83 & 109.83 & 109.83 & 109.83 \\
\cline{2-11}
 & \multirow[c]{4}{*}{10} & 0 & 91.15 & 91.15 & 91.15 & 91.15 & 91.15 & 91.15 & 91.15 & 91.15 \\
 &  & 30 & 83.75 & 83.75 & 83.75 & 83.75 & 83.75 & 83.75 & 83.75 & 83.75 \\
 &  & 60 & 82.16 & 82.16 & 82.16 & 82.16 & 82.16 & 82.16 & 82.16 & 82.16 \\
 &  & 90 & 94.79 & 94.79 & 94.79 & 94.79 & 94.79 & 94.79 & 94.79 & 94.79 \\
\cline{1-11} \cline{2-11}
\end{tabular}
\end{table}
\begin{table}
\tiny
\caption{Final lower bounds (lb) and upper bounds (ub) on grid instances (continuation)}
\label{tab:results}
\begin{tabular}{lllrrrrrrrr}
\toprule
 &  &  & \multicolumn{2}{c}{\flow} & \multicolumn{2}{c}{\tree} & \multicolumn{2}{c}{\ERS} & \multicolumn{2}{c}{\ERSTree} \\
 &  &  & lb & ub & lb & ub & lb & ub & lb & ub \\
Size & $k$ & $\rho$ &  &  &  &  &  &  &  &  \\
\midrule
\multirow[c]{32}{*}{$6 \times 6$} & \multirow[c]{4}{*}{3} & 0 & 1349.24 & 1349.31 & 509.44 & 1418.51 & 1349.31 & 1349.31 & 1349.31 & 1349.31 \\
 &  & 30 & 1767.60 & 1767.77 & 686.78 & 1903.80 & 1767.77 & 1767.77 & 1767.77 & 1767.77 \\
 &  & 60 & 1774.35 & 1774.52 & 792.20 & 1902.43 & 1774.52 & 1774.52 & 1774.52 & 1774.52 \\
 &  & 90 & 1641.41 & 1641.43 & 707.76 & 1656.83 & 1641.43 & 1641.43 & 1641.43 & 1641.43 \\
\cline{2-11}
 & \multirow[c]{4}{*}{4} & 0 & 629.75 & 939.96 & 466.40 & 939.96 & 939.96 & 939.96 & 939.96 & 939.96 \\
 &  & 30 & 1025.63 & 1025.72 & 692.83 & 1025.72 & 1025.72 & 1025.72 & 1025.62 & 1025.72 \\
 &  & 60 & 1013.97 & 1232.06 & 721.65 & 1274.08 & 1232.06 & 1232.06 & 1231.99 & 1232.06 \\
 &  & 90 & 1039.54 & 1039.64 & 631.82 & 1083.31 & 1039.64 & 1039.64 & 1039.64 & 1039.64 \\
\cline{2-11}
 & \multirow[c]{4}{*}{5} & 0 & 372.53 & 654.37 & 455.74 & 654.37 & 654.37 & 654.37 & 654.37 & 654.37 \\
 &  & 30 & 543.07 & 761.54 & 564.81 & 761.54 & 761.54 & 761.54 & 761.54 & 761.54 \\
 &  & 60 & 571.54 & 911.69 & 679.60 & 911.69 & 911.69 & 911.69 & 911.69 & 911.69 \\
 &  & 90 & 668.43 & 668.48 & 544.30 & 668.48 & 668.48 & 668.48 & 668.48 & 668.48 \\
\cline{2-11}
 & \multirow[c]{4}{*}{6} & 0 & 242.17 & 486.66 & 421.92 & 486.66 & 486.66 & 486.66 & 486.66 & 486.66 \\
 &  & 30 & 323.77 & 615.52 & 615.41 & 615.47 & 615.47 & 615.47 & 615.47 & 615.47 \\
 &  & 60 & 314.78 & 615.38 & 615.32 & 615.38 & 615.38 & 615.38 & 615.38 & 615.38 \\
 &  & 90 & 355.68 & 481.95 & 481.91 & 481.95 & 481.95 & 481.95 & 481.93 & 481.95 \\
\cline{2-11}
 & \multirow[c]{4}{*}{7} & 0 & 151.66 & 423.68 & 387.38 & 387.41 & 387.41 & 387.41 & 387.41 & 387.41 \\
 &  & 30 & 207.65 & 501.98 & 501.93 & 501.98 & 501.98 & 501.98 & 501.98 & 501.98 \\
 &  & 60 & 203.92 & 453.59 & 453.59 & 453.59 & 453.59 & 453.59 & 453.56 & 453.59 \\
 &  & 90 & 225.04 & 377.45 & 377.42 & 377.45 & 377.45 & 377.45 & 377.45 & 377.45 \\
\cline{2-11}
 & \multirow[c]{4}{*}{8} & 0 & 109.49 & 324.43 & 324.40 & 324.43 & 324.43 & 324.43 & 324.43 & 324.43 \\
 &  & 30 & 148.85 & 411.28 & 409.16 & 409.18 & 409.18 & 409.18 & 409.18 & 409.18 \\
 &  & 60 & 190.61 & 367.38 & 367.38 & 367.38 & 367.38 & 367.38 & 367.38 & 367.38 \\
 &  & 90 & 176.44 & 297.93 & 297.93 & 297.93 & 297.93 & 297.93 & 297.93 & 297.93 \\
\cline{2-11}
 & \multirow[c]{4}{*}{9} & 0 & 68.12 & 270.65 & 270.64 & 270.65 & 270.65 & 270.65 & 270.65 & 270.65 \\
 &  & 30 & 93.83 & 327.55 & 327.55 & 327.55 & 327.55 & 327.55 & 327.55 & 327.55 \\
 &  & 60 & 103.24 & 308.87 & 308.87 & 308.87 & 308.87 & 308.87 & 308.87 & 308.87 \\
 &  & 90 & 125.68 & 237.36 & 237.36 & 237.36 & 237.36 & 237.36 & 237.36 & 237.36 \\
\cline{2-11}
 & \multirow[c]{4}{*}{10} & 0 & 61.97 & 220.95 & 220.93 & 220.95 & 220.95 & 220.95 & 220.95 & 220.95 \\
 &  & 30 & 71.10 & 264.41 & 264.41 & 264.41 & 264.41 & 264.41 & 264.41 & 264.41 \\
 &  & 60 & 103.14 & 252.38 & 252.38 & 252.38 & 252.38 & 252.38 & 252.38 & 252.38 \\
 &  & 90 & 113.31 & 191.79 & 191.79 & 191.79 & 191.79 & 191.79 & 191.79 & 191.79 \\
\cline{1-11} \cline{2-11}
\multirow[c]{32}{*}{$7 \times 7$} & \multirow[c]{4}{*}{3} & 0 & 1678.47 & 2769.95 & 689.07 & 2774.52 & 2529.33 & 2774.52 & 2592.46 & 2763.55 \\
 &  & 30 & 2030.43 & 3268.98 & 806.48 & $\infty$ & 2926.32 & 3116.57 & 2823.13 & 3116.57 \\
 &  & 60 & 2063.96 & 2439.55 & 646.96 & 2753.08 & 2439.55 & 2439.55 & 2439.55 & 2439.55 \\
 &  & 90 & 3391.85 & 4567.88 & 1038.94 & 4614.88 & 4567.88 & 4567.88 & 4567.87 & 4567.88 \\
\cline{2-11}
 & \multirow[c]{4}{*}{4} & 0 & 754.37 & 1889.76 & 675.49 & 1913.77 & 1758.35 & 1874.74 & 1817.14 & 1913.77 \\
 &  & 30 & 804.96 & 2233.78 & 760.91 & 2313.05 & 2042.02 & 2094.72 & 2024.10 & 2094.72 \\
 &  & 60 & 1051.56 & 1680.58 & 598.08 & 2098.96 & 1680.58 & 1680.58 & 1680.58 & 1680.58 \\
 &  & 90 & 1781.99 & 2981.79 & 990.86 & 3092.89 & 2870.83 & 2981.79 & 2981.57 & 2981.79 \\
\cline{2-11}
 & \multirow[c]{4}{*}{5} & 0 & 477.46 & 1430.04 & 620.46 & 1445.17 & 1332.66 & 1406.28 & 1406.28 & 1406.28 \\
 &  & 30 & 461.58 & 1617.61 & 718.62 & 1692.61 & 1493.07 & 1617.61 & 1514.24 & 1617.61 \\
 &  & 60 & 520.64 & 1210.27 & 558.25 & 1220.54 & 1210.27 & 1210.27 & 1210.27 & 1210.27 \\
 &  & 90 & 892.98 & 2148.17 & 916.64 & 2148.17 & 1992.35 & 2196.46 & 2048.94 & 2148.17 \\
\cline{2-11}
 & \multirow[c]{4}{*}{6} & 0 & 224.03 & 1103.87 & 583.01 & 1177.74 & 1079.68 & 1079.70 & 1079.67 & 1079.70 \\
 &  & 30 & 185.02 & 1305.52 & 603.99 & 1277.96 & 1200.52 & 1279.93 & 1236.28 & 1279.93 \\
 &  & 60 & 296.20 & 988.35 & 520.33 & 988.29 & 958.54 & 972.95 & 972.95 & 972.95 \\
 &  & 90 & 430.95 & 1592.27 & 826.40 & 1835.29 & 1534.47 & 1592.27 & 1592.27 & 1592.27 \\
\cline{2-11}
 & \multirow[c]{4}{*}{7} & 0 & 132.23 & 906.01 & 559.22 & 906.01 & 874.33 & 906.01 & 906.01 & 906.01 \\
 &  & 30 & 132.16 & 1088.90 & 618.45 & 1070.98 & 967.91 & 1049.79 & 1009.12 & 1048.88 \\
 &  & 60 & 145.36 & 830.85 & 475.70 & 826.08 & 768.19 & 803.74 & 803.74 & 803.74 \\
 &  & 90 & 283.12 & 1241.79 & 780.94 & 1241.79 & 1241.79 & 1241.79 & 1241.79 & 1241.79 \\
\cline{2-11}
 & \multirow[c]{4}{*}{8} & 0 & 83.81 & 763.12 & 513.14 & 756.65 & 736.39 & 756.65 & 756.65 & 756.65 \\
 &  & 30 & 79.86 & 901.62 & 597.46 & 872.52 & 826.15 & 856.68 & 856.60 & 856.68 \\
 &  & 60 & 74.75 & 672.98 & 456.81 & 676.05 & 652.33 & 652.33 & 652.33 & 652.33 \\
 &  & 90 & 160.54 & 1061.00 & 741.86 & 1023.87 & 1023.87 & 1023.87 & 1023.87 & 1023.87 \\
\cline{2-11}
 & \multirow[c]{4}{*}{9} & 0 & 59.01 & 617.86 & 472.18 & 617.86 & 617.86 & 617.86 & 617.86 & 617.86 \\
 &  & 30 & 56.68 & 695.87 & 547.50 & 695.87 & 695.85 & 695.87 & 695.87 & 695.87 \\
 &  & 60 & 64.36 & 558.77 & 425.20 & 558.77 & 558.77 & 558.77 & 558.77 & 558.77 \\
 &  & 90 & 118.54 & 886.26 & 664.82 & 840.10 & 840.04 & 840.10 & 840.10 & 840.10 \\
\cline{2-11}
 & \multirow[c]{4}{*}{10} & 0 & 45.02 & 535.40 & 453.95 & 535.40 & 535.40 & 535.40 & 535.40 & 535.40 \\
 &  & 30 & 42.72 & 599.21 & 512.51 & 599.21 & 594.71 & 594.75 & 594.75 & 594.75 \\
 &  & 60 & 50.40 & 527.08 & 381.28 & 482.36 & 482.36 & 482.36 & 482.36 & 482.36 \\
 &  & 90 & 85.29 & 716.97 & 646.34 & 702.49 & 702.49 & 702.49 & 702.49 & 702.49 \\
\cline{1-11} \cline{2-11}
\end{tabular}
\end{table}

\begin{table}
\tiny
\caption{Final lower bounds (lb) and upper bounds (ub) on grid instances (continuation)}
\label{tab:results}
\begin{tabular}{lllrrrrrrrr}
\toprule
 &  &  & \multicolumn{2}{c}{\flow} & \multicolumn{2}{c}{\tree} & \multicolumn{2}{c}{\ERS} & \multicolumn{2}{c}{\ERSTree} \\
 &  &  & lb & ub & lb & ub & lb & ub & lb & ub \\
Size & $k$ & $\rho$ &  &  &  &  &  &  &  &  \\
\midrule
\multirow[c]{32}{*}{$8 \times 8$} & \multirow[c]{4}{*}{3} & 0 & 2103.43 & 5973.56 & 1125.15 & $\infty$ & 4911.10 & 6003.52 & 4985.74 & 7052.40 \\
 &  & 30 & 2540.03 & 6912.07 & 1282.96 & $\infty$ & 5281.30 & 6423.68 & 5332.67 & 6276.32 \\
 &  & 60 & 2470.59 & 7856.03 & 1156.25 & $\infty$ & 5715.57 & 8635.63 & 5734.72 & 9076.03 \\
 &  & 90 & 3466.48 & 7030.71 & 1022.57 & $\infty$ & 5801.75 & 10574.24 & 5908.83 & 8205.28 \\
\cline{2-11}
 & \multirow[c]{4}{*}{4} & 0 & 1278.90 & 4298.46 & 1079.08 & $\infty$ & 3522.21 & 4450.32 & 3569.06 & 6748.73 \\
 &  & 30 & 729.08 & 4805.46 & 1202.77 & 7095.49 & 3771.29 & 4335.87 & 3801.84 & 4202.35 \\
 &  & 60 & 960.00 & 5561.16 & 1098.74 & $\infty$ & 4037.15 & 5010.12 & 4059.69 & 5082.09 \\
 &  & 90 & 1489.67 & 4982.64 & 994.07 & $\infty$ & 3872.83 & 4915.32 & 3958.06 & 4858.57 \\
\cline{2-11}
 & \multirow[c]{4}{*}{5} & 0 & 462.73 & 3739.51 & 1008.03 & $\infty$ & 2709.32 & 3063.47 & 2784.36 & 4012.07 \\
 &  & 30 & 416.54 & 3462.78 & 1132.74 & $\infty$ & 2891.03 & 3312.89 & 2976.16 & 3325.59 \\
 &  & 60 & 527.12 & 3624.13 & 996.81 & $\infty$ & 3011.62 & 3875.09 & 3078.19 & 3663.09 \\
 &  & 90 & 583.57 & 3540.83 & 897.37 & $\infty$ & 2882.01 & 3353.56 & 2981.67 & 3540.83 \\
\cline{2-11}
 & \multirow[c]{4}{*}{6} & 0 & 285.05 & 2600.64 & 974.76 & 2832.56 & 2166.55 & 2551.45 & 2264.48 & 2406.60 \\
 &  & 30 & 268.90 & 3154.87 & 1095.77 & 2976.64 & 2347.07 & 2612.89 & 2436.71 & 2676.97 \\
 &  & 60 & 268.50 & 2755.93 & 1007.68 & 3408.67 & 2367.42 & 2792.89 & 2459.39 & 2707.78 \\
 &  & 90 & 398.11 & 2699.84 & 838.09 & $\infty$ & 2218.24 & 2646.57 & 2338.62 & 2717.08 \\
\cline{2-11}
 & \multirow[c]{4}{*}{7} & 0 & 123.85 & 2131.64 & 919.74 & 2685.26 & 1795.60 & 2008.06 & 1876.03 & 2011.47 \\
 &  & 30 & 99.13 & 2296.55 & 1067.21 & 2319.46 & 1931.97 & 2230.28 & 2037.86 & 2173.70 \\
 &  & 60 & 152.99 & 2455.03 & 967.01 & 2413.04 & 1929.24 & 2117.70 & 2018.97 & 2171.89 \\
 &  & 90 & 169.09 & 2140.64 & 838.93 & 2236.05 & 1775.06 & 2196.16 & 1891.55 & 2140.85 \\
\cline{2-11}
 & \multirow[c]{4}{*}{8} & 0 & 97.88 & 1693.65 & 863.66 & 1759.77 & 1519.91 & 1690.24 & 1596.82 & 1699.47 \\
 &  & 30 & 74.21 & 1806.42 & 1019.18 & 2326.53 & 1647.51 & 1806.42 & 1734.91 & 1806.42 \\
 &  & 60 & 94.22 & 1925.83 & 909.31 & 1940.40 & 1598.69 & 1811.76 & 1680.71 & 1811.76 \\
 &  & 90 & 113.69 & 1727.89 & 800.01 & 1964.64 & 1474.12 & 1698.50 & 1573.11 & 1698.10 \\
\cline{2-11}
 & \multirow[c]{4}{*}{9} & 0 & 48.77 & 1446.53 & 846.94 & 1438.15 & 1298.14 & 1443.12 & 1372.40 & 1446.53 \\
 &  & 30 & 52.94 & 1521.38 & 955.93 & 1609.04 & 1425.55 & 1496.87 & 1496.87 & 1496.87 \\
 &  & 60 & 81.27 & 1537.34 & 878.10 & 1634.03 & 1353.15 & 1540.06 & 1429.38 & 1537.34 \\
 &  & 90 & 80.31 & 1426.21 & 755.83 & 1462.72 & 1239.85 & 1443.70 & 1329.74 & 1422.37 \\
\cline{2-11}
 & \multirow[c]{4}{*}{10} & 0 & 36.44 & 1207.75 & 803.10 & 1225.78 & 1130.85 & 1218.52 & 1190.21 & 1207.75 \\
 &  & 30 & 34.99 & 1388.10 & 902.24 & 1390.80 & 1246.62 & 1332.66 & 1316.56 & 1332.66 \\
 &  & 60 & 43.10 & 1326.57 & 794.38 & 1348.15 & 1161.17 & 1310.56 & 1244.06 & 1312.02 \\
 &  & 90 & 53.21 & 1163.47 & 713.38 & 1191.54 & 1067.71 & 1163.47 & 1132.24 & 1163.47 \\
\cline{1-11} \cline{2-11}
\multirow[c]{32}{*}{$9 \times 9$} & \multirow[c]{4}{*}{3} & 0 & 2306.44 & 8368.85 & 1158.35 & $\infty$ & 6320.35 & 29126.32 & 6370.68 & 27711.16 \\
 &  & 30 & 2408.97 & 8864.60 & 1086.03 & $\infty$ & 6516.67 & 29818.35 & 6579.82 & 31867.45 \\
 &  & 60 & 2744.19 & 9147.63 & 1155.30 & $\infty$ & 6795.59 & 32752.43 & 6969.17 & 34800.02 \\
 &  & 90 & 4502.30 & 11197.40 & 1457.40 & $\infty$ & 10211.92 & 60780.37 & 10231.97 & 57636.52 \\
\cline{2-11}
 & \multirow[c]{4}{*}{4} & 0 & 1456.40 & 6570.81 & 1135.46 & $\infty$ & 4551.91 & 26010.00 & 4586.93 & 26422.88 \\
 &  & 30 & 941.11 & 6342.96 & 1120.08 & $\infty$ & 4674.83 & 23499.13 & 4725.97 & 26771.57 \\
 &  & 60 & 846.08 & 6703.23 & 1058.46 & $\infty$ & 4861.11 & 32159.55 & 4986.72 & 30773.68 \\
 &  & 90 & 1628.75 & 8153.97 & 1338.00 & $\infty$ & 6786.19 & 9518.38 & 6846.62 & 9044.52 \\
\cline{2-11}
 & \multirow[c]{4}{*}{5} & 0 & 371.61 & 5012.56 & 1063.89 & $\infty$ & 3553.11 & 15336.98 & 3611.80 & 23902.28 \\
 &  & 30 & 448.31 & 5013.70 & 1073.81 & $\infty$ & 3639.32 & 30084.80 & 3710.21 & 24418.90 \\
 &  & 60 & 362.08 & 4832.50 & 1046.23 & $\infty$ & 3731.82 & 28807.57 & 3841.59 & 28023.08 \\
 &  & 90 & 817.84 & 6179.81 & 1290.64 & $\infty$ & 5043.04 & 42110.65 & 5143.40 & 6083.49 \\
\cline{2-11}
 & \multirow[c]{4}{*}{6} & 0 & 215.78 & 4118.69 & 1013.51 & $\infty$ & 2902.78 & 7612.32 & 2952.78 & 4427.37 \\
 &  & 30 & 236.04 & 3756.49 & 1045.29 & $\infty$ & 2951.97 & 27465.70 & 3034.77 & 19550.18 \\
 &  & 60 & 228.91 & 3961.43 & 1033.77 & $\infty$ & 2982.31 & 31384.41 & 3105.96 & 13480.52 \\
 &  & 90 & 464.67 & 5034.68 & 1224.86 & $\infty$ & 3984.28 & 5190.79 & 4079.92 & 5664.49 \\
\cline{2-11}
 & \multirow[c]{4}{*}{7} & 0 & 74.67 & 2970.54 & 1008.40 & $\infty$ & 2436.11 & 24281.70 & 2488.21 & 15585.68 \\
 &  & 30 & 84.75 & 3473.85 & 957.91 & $\infty$ & 2465.00 & 11101.72 & 2524.54 & 6057.12 \\
 &  & 60 & 75.94 & 3501.08 & 963.40 & $\infty$ & 2465.70 & 31861.58 & 2575.71 & 6057.45 \\
 &  & 90 & 186.56 & 4254.35 & 1258.62 & 5320.46 & 3290.45 & 4823.05 & 3374.45 & 7247.76 \\
\cline{2-11}
 & \multirow[c]{4}{*}{8} & 0 & 48.37 & 2557.55 & 982.18 & $\infty$ & 2085.04 & 17982.80 & 2126.68 & 2491.31 \\
 &  & 30 & 58.96 & 3061.32 & 977.16 & $\infty$ & 2100.76 & 24416.65 & 2176.46 & 2865.63 \\
 &  & 60 & 46.37 & 2430.62 & 893.37 & 3013.63 & 2074.67 & 26031.13 & 2169.08 & 2722.66 \\
 &  & 90 & 87.05 & 3327.22 & 1220.81 & 4257.36 & 2746.90 & 3527.88 & 2835.54 & 3794.86 \\
\cline{2-11}
 & \multirow[c]{4}{*}{9} & 0 & 31.14 & 2344.15 & 820.35 & $\infty$ & 1816.71 & 17531.22 & 1860.33 & 2257.66 \\
 &  & 30 & 42.90 & 2458.06 & 929.79 & 2776.75 & 1815.23 & 2798.69 & 1876.45 & 3003.00 \\
 &  & 60 & 45.13 & 2213.88 & 884.38 & 2776.22 & 1791.45 & 2030.41 & 1875.80 & 2027.14 \\
 &  & 90 & 69.66 & 2809.89 & 1167.44 & 3072.89 & 2334.24 & 2949.30 & 2435.06 & 2763.84 \\
\cline{2-11}
 & \multirow[c]{4}{*}{10} & 0 & 24.87 & 1845.03 & 910.05 & 2322.44 & 1598.02 & 2927.63 & 1641.86 & 2516.69 \\
 &  & 30 & 32.94 & 1896.25 & 870.71 & 1967.03 & 1591.13 & 1971.83 & 1663.10 & 1980.77 \\
 &  & 60 & 27.64 & 1919.38 & 846.95 & 2037.87 & 1545.66 & 2211.21 & 1643.27 & 1798.90 \\
 &  & 90 & 52.31 & 2563.98 & 1127.31 & 2913.75 & 2003.98 & 2188.39 & 2100.26 & 2188.39 \\
\cline{1-11} \cline{2-11}
\bottomrule
\end{tabular}
\end{table}

\end{document}